\documentclass[conference]{IEEEtran}
\IEEEoverridecommandlockouts
\usepackage{cite}
\usepackage{algorithmic}
\usepackage{graphicx}
\usepackage{textcomp}
\usepackage{xcolor}
\usepackage[all]{nowidow}
\usepackage{setspace}
\def\BibTeX{{\rm B\kern-.05em{\sc i\kern-.025em b}\kern-.08em
    T\kern-.1667em\lower.7ex\hbox{E}\kern-.125emX}}

\usepackage{balance}
\usepackage{enumitem}
\usepackage{amsmath,amssymb,amsfonts,amsthm}
\usepackage[capitalize,nameinlink]{cleveref}

\newcommand{\ignore}[1]{}

\ignore{
\titleformat*{\subsection}{\fontsize{11}{12} \selectfont\bfseries}
\titleformat*{\subsubsection}{\bfseries}
\titlespacing*{\subsection}{0pt}{1.75ex plus 0.5ex minus 0.5ex}{0.5ex plus .25ex minus .25ex}
}

\newcommand{\eg}{{e.g.}}
\newcommand{\ie}{{i.e.}}
\newcommand{\etc}{{etc.}}

\renewcommand{\S}{Section}

\renewcommand{\paragraph}[1]{\noindent{\textbf{#1:}}}
\newcommand{\comment}[1]{}
\newcommand{\arxivonly}[1]{#1}

\newcommand{\trusts}{\ensuremath{\textit{\, trusts} \,\,}}
\newcommand{\sattestor}[1]{\ensuremath{\mathsf{sattestor}#1}}
\newcommand{\sattestordel}[1]{\ensuremath{\mathsf{sattestor}^{*}#1}}

\newcommand{\satt}[1]{\ensuremath{\mathsf{sattestor}(#1)}}
\newcommand{\sattdel}[1]{\ensuremath{\mathsf{sattestor}^{*}(#1)}}
\newcommand{\sattboth}[1]{\ensuremath{\mathsf{sattestor}^{(*)}(#1)}}
\newcommand{\bound}[1]{\ensuremath{\mathsf{bound}#1}}

\newcommand{\mylabels}[1]{\ensuremath{\mathsf{labels}(#1)}}
\newcommand{\sattfree}[1]{\ensuremath{\mathsf{sattfree}(#1)}}
\newcommand{\mynames}[1]{\ensuremath{\mathsf{IDs}(#1)}}
\newcommand{\blabel}[1]{\ensuremath{\mathsf{blabel}(#1)}}
\newcommand{\blabels}[1]{\ensuremath{\mathsf{blabels}(#1)}}
\newcommand{\says}{\ensuremath{\, \textit{says} \,\,}}

\newcommand{\domain}[1]{\ensuremath{D_{#1}}}
\newcommand{\saname}[1]{\ensuremath{O_{#1}}}

\newcommand{\satlabel}[1]{\ensuremath{\ell_{#1}}}

\newcommand{\sflabelset}[1]{\ensuremath{\Lambda_{#1}}}
\newcommand{\anylabelset}[1]{\ensuremath{\Gamma_{#1}}}
\newcommand{\allsflabels}{\ensuremath{\widehat{\Lambda}}}
\newcommand{\alllabels}{\ensuremath{\widehat{\Gamma}}}

\newcommand{\upair}[2]{\ensuremath{\left(#1,#2\right)}}
\newcommand{\tpair}[2]{\ensuremath{\left<#1,#2\right>}}
\newcommand{\ltriple}[3]{\ensuremath{\left<#1,#2,#3\right>}}

\newcommand{\lasdel}[2]{\ensuremath{\mathsf{DEL}(#1,#2)}}
\newcommand{\lassatt}[2]{\ensuremath{\mathsf{SATT}(#1,#2)}}
\newcommand{\lassattdel}[2]{\ensuremath{\mathsf{SATT}^{*}(#1,#2)}}

\newcommand{\lleq}{\ensuremath{\succeq}}
\newcommand{\llt}{\ensuremath{\succ}}
\newcommand{\lgeq}{\ensuremath{\preceq}}
\newcommand{\lgt}{\ensuremath{\prec}}
\newcommand{\wkbroadens}{\lleq}
\newcommand{\broadens}{\llt}
\newcommand{\wknarrows}{\lgeq}
\newcommand{\narrows}{\lgt}
\newcommand{\mcond}{\rightarrow}
\def\semcons{\, |\kern-.3em\models\,} 
\catcode`\@=11 
\def\notvdash{\mathrel{\mathpalette\c@ncel\vdash}}
\def\notmodels{\mathrel{\mathpalette\c@ncel\models}} 
\catcode`\@=12

\newcommand{\itrust}[1]{\ensuremath{\mathsf{T}_{#1}^0}}
\newcommand{\las}[1]{\ensuremath{\mathsf{LAS}_{#1}}}
\newcommand{\tstmts}{\ensuremath{\mathsf{TS}}}
\newcommand{\sata}[1]{\ensuremath{(\domain{#1},\saname{#1})}}
\newcommand{\bsata}[1]{\ensuremath{\langle\domain{#1},\saname{#1}\rangle}}
\newcommand{\blsata}[1]{\ensuremath{\langle\domain{#1},\saname{#1},\satlabel{#1}\rangle}}

\newcommand{\bsinglesata}[1]{\ensuremath{\langle\domain{#1},\saname{#1},\satlabel{#1}\rangle}}

\newcommand{\maybe}[1]{\ensuremath{\textit{may-be}_{#1}}}

\newcommand{\present}{\ensuremath{\textit{Present}}}

\newcommand{\resulta}{\ensuremath{\textit{Resulta}}}
\newcommand{\resultb}{\ensuremath{\textit{Resultb}}}
\newcommand{\past}{\ensuremath{\textit{Past}}}
\newcommand{\rest}{\ensuremath{\textit{Rest}}}

\newcommand{\maxlim}[2]{\ensuremath{\mathsf{maxlim}(#1#2)}}
\newcommand{\lc}[1]{\ensuremath{\mathsf{lc}(#1)}}
\newcommand{\lalc}[2]{\ensuremath{\mathsf{lc}(#1,#2)}}

\newcommand{\dlen}[1]{\ensuremath{\mathsf{mlen}(#1)}}

\newcommand{\Der}{\ensuremath{\textit{Der}}}

\newtheorem{theorem}{Theorem}[section]
\newtheorem{lemma}[theorem]{Lemma}
\newtheorem{proposition}[theorem]{Proposition}

\theoremstyle{definition}
\newtheorem{definition}[theorem]{Definition}

\theoremstyle{remark}
\newtheorem{q}{Question}

\newtheorem{example}[q]{Example}

\author{
\IEEEauthorblockN{Aaron D. Jaggard}
\IEEEauthorblockA{U.S. Naval Research Laboratory\\
Washington DC, USA \\
aaron.d.jaggard.civ@us.navy.mil}
\and
\IEEEauthorblockN{Paul Syverson}
\IEEEauthorblockA{U.S. Naval Research Laboratory\\
Washington DC, USA \\
paul.f.syverson.civ@us.navy.mil}
\and
\IEEEauthorblockN{Catherine Meadows}
\IEEEauthorblockA{U.S. Naval Research Laboratory\\
Washington DC, USA \\
catherine.a.meadows2.civ@us.navy.mil}
}

\title{A Logic of Sattestation}

\begin{document}

\maketitle

%% Footer for distribution statement %%
%% Option 1: Fancyhdr places the statement in the bottom margin
% \usepackage{fancyhdr}
% %\pagestyle{fancy} % all pages
% \thispagestyle{fancy} % first page only
% \fancyhead{} % clear default header config
% % Custom footer
% \newcommand\hdrstyle[1]{{{\footnotesize #1}}}
% \fancyfoot[L]{\hdrstyle{Distribution A: approved for public release; distribution is unlimited.}}
% \renewcommand{\headrulewidth}{0.0pt} % clear line above header
%% End Option 1 %%
%% Option 2: just use a regular footnote, but don't assign it a number.
\newcommand\distfootnote[1]{%
  \begingroup
  \renewcommand\thefootnote{}\footnote{#1}%
  \addtocounter{footnote}{-1}%
  \endgroup
}
%% End Option 2 %%
%% End Footer for distribution statement %%

\begin{abstract}
We introduce a logic for reasoning about contextual trust for web
addresses, provide a Kripke semantics for it, and prove its soundness under
reasonable assumptions about principals' policies.

Self-Authenticating Traditional Addresses (SATAs) are valid DNS
addresses or URLs that are generally meaningful---to both humans and
web infrastructure---and contain a commitment to a public key in the
address itself. Trust in web addresses is currently established via
domain name registration, TLS certificates, and other hierarchical
elements of the internet infrastructure. SATAs support such structural
roots of trust but also complementary contextual roots associated with
descriptive properties. The existing structural roots leave web
connections open to a variety of well-documented and significant
hijack vulnerabilities.  Contextual trust roots provide, among other
things, stronger resistance to such vulnerabilities.

We also consider labeled SATAs, which include descriptive properties
such as that a SATA is an address for a news organization, a site
belonging to a particular government or company, a site with
information about a certain topic, etc.  Our logic addresses both
trust in the bound together identity of the address and trust in the binding
of labels to it.
Our logic allows reasoning about delegation of trust with respect to
specified labels, relationships between labels that provide more or
less specific information, and the interaction between these two
aspects.

In addition to soundness, we prove that if a principal trusts a
particular identity (possibly with label), then either this trust is
initially assumed, or there is a trust chain of delegations to this
from initial trust assumptions. We also present an algorithm that
effectively derives all possible trust statements from the set of
initial trust assumptions and show it to be sound, complete, and
terminating.\distfootnote{This is a full version, including proofs, of a paper to appear in the \emph{Proceedings of the 37$^\mathrm{th}$ IEEE Computer Security Foundations Symposium} (CSF'24).\\
Distribution Statement A.  Approved for public release: distribution is unlimited.}

\end{abstract}

\section{Introduction}
\label{sec:introduction}

We introduce a logic for reasoning about contextual trust for web
addresses, both for properties essential to the identity of the
address and for descriptive properties.  These properties, or labels,
say things such as: that an address is
for a news organization; that it belongs to the U.S. Government, Microsoft, or other specified entity; that it relates to a topic like cooking; etc.  Importantly, these are descriptive and not for authorization.  Trust in web addresses is currently established via
domain name registration, TLS certificates, and other hierarchical
components of the internet infrastructure. Formal reasoning about
trust in TLS certificates (or predecessors of them) has long been
researched~\cite{gaarder1991applying} as has reasoning about
alternatives to hierarchies like X.509 via local naming such as SPKI
or SDSI~\cite{abadi1998sdsi}. In our logic, names (addresses) are
always global so as to remain entirely compatible with the predominant
existing internet naming and authentication infrastructure.  It is the
descriptive properties of them that can be local or contextual.  But
we also reason about contextual roots of trust complementary to the
existing web authentication infrastructure both for global addresses
and for descriptive properties of them.

Traditional domain names like \texttt{\small microsoft.com} are meaningful:
users understand that connecting to \texttt{\small https://www.microsoft.com}
will take them to a website offering information about products and
services offered by Microsoft. They are also meaningful to
the internet's infrastructure: \texttt{\small microsoft.com} is
meaningful to DNS, which is generally able to provide an IP address
for it.  Further, there is structural meaning to domain names understood
by both humans and infrastructure systems: \texttt{support.microsoft.com}
is understood to be a subdomain of \texttt{\small microsoft.com}, a subdomain
at which users might expect to find support information for Microsoft
products and services.

To protect against hijack and other attacks, TLS certificates signed
by Certificate Authorities (CAs) are understood by web browsers.
These assert that the key authenticating a connection to
\texttt{\small https://www.microsoft.com} is indeed properly bound to that
domain. In contrast, an \emph{onion address} in Tor needs no certificate to assert a
binding between the authenticating key and the connection's
address. It is \emph{self-authenticating}: the onion address is a
56-character string encoding the public key for the
private key used to authenticate the
connection~\cite{ng-onion-services-224}. However, onion addresses are not
meaningful. For humans they appear to be long random strings
of letters and numbers; the RFC that recognizes \texttt{\small .onion} as
a reserved TLD also prohibits DNS resolution of onion
addresses~\cite{ietf-onion-tld-rfc}; and most popular browsers are not
set up to use the onion-service protocols needed for address
resolution or reaching an onion address or, more generally, for
connections via the Tor network.

Self-authenticating traditional addresses (SATAs) were introduced to
combine the meaningfulness of traditional web addresses to both humans
and the common infrastructure (e.g., DNS and popular browsers) with
the self-authentication of onion
addresses~\cite{once-and-future,secdev19,satas-wpes21}. A SATA for the
Tor Project would combine its traditional domain name,
\texttt{\small www.torproject.org}, with its onion address, thus for example:\\
\texttt{\small www.torproject.org/?onion=2gzyxa5ihm7nsggfxnu52 rck2vv4rvmdlkiu3zzui5du4xyclen53wid}.

SATAs do not by themselves bring self-authentication to meaningful
domain names. Seeing the above SATA in her browser, how could Alice
know whether or not an attacker with a hijacked certificate for
\texttt{\small www.torproject.org}, had associated its own onion
address versus one actually belonging to the Tor Project? Prior work
that we will briefly review presently sets out the technical means for
binding the parts of a SATA\@. Part of the work of this paper is to
provide a logic for reasoning about how Alice can ground trusting that
such a SATA is bound together in trust credentials that are based on
meaningful associations rather than just structural authority.

The above SATA format has usable security and deployment
advantages~\cite{satas-wpes21}, but the first introduced format was
for the onion address to be a subdomain of the traditional domain
name~\cite{once-and-future}. This allowed the onion address to be
bound in a Domain Validation (DV) certificate by a CA using then
existing industry standards for certificate issuance. And that format
is still used in this way, further making onion addresses meaningful
to the internet infrastructure, in this case the infrastructure for
TLS certificate issuance. The certificate indicates that the CA has
done a standard check that someone controlling the traditional domain
name intends the binding of the onion address to that domain
name. Since onion addresses are self-authenticating, a SATA owner can
attest that someone controlling the onion address also intends this
binding by signing an HTTP header also containing both addresses,
along with a hash of the TLS certificate, a validity interval,
etc. This is called a \emph{self-sattestation}.

If a full SATA is known, such as by being recalled from a browser
bookmark, then this might be sufficient. But if address discovery is
via the traditional domain name, such as from user memory or via a
search engine lookup of, e.g., ``microsoft'' or ``New York Times'',
then it is not.  Even with self-sattestation, SATA authentication as
described so far would only be as strong as the structural guarantees
of TLS certificate issuance: an attacker able to obtain a fraudulent
TLS certificate for a domain could assert binding of an onion address
under their control within that certificate.  Since onion addresses
are not easily recognized, this could pass undetected. And TLS
certificates are subject to significant hijack
risk~\cite{tls-hijack-ars} despite substantial efforts to reduce that
risk (e.g., via DNSSEC or Extended Validation (EV) certificates) or
to detect misissuance (e.g., via Certificate Transparency (CT) logs).

For this reason, one SATA can provide a \emph{sattestation} for
another.  An example suggested in~\cite{secdev19} would be for a SATA
controlled by the Freedom of the Press Foundation or the Berkman Klein
Center for Internet \& Society to provide a sattestatation for a SATA
for CNN that the components comprising CNN's SATA are properly bound
together. As introduced in~\cite{satas-wpes21}, sattestations are
implemented as JSON formatted headers signed using the onion key of
the sattestor SATA and may contain descriptive labels, such as that
the sattestee is a ``news'' site. Worked examples of sattestations and
showing such JSON headers are provided in
Appendix~\ref{sec:appendix-worked-example}.

Prior work described implementation and formats of sattestation
but not analysis or formalization for reasoning,
neither about contextual trust for sattestors asserting the binding of
a SATA's identity nor for the structure of related descriptive
labels in sattestations. These are set out in
Section~\ref{sec:normal-modal-logic}.

\subsection{Contributions}
\label{sec:contributions}

In this paper we

\begin{itemize}

\item introduce concepts of identity and contextual properties for
  web addresses and an associated concept of trusted belief in the
  binding together of such identities with or without associated
  properties. [Section~\ref{sec:identities}]
  
\item introduce a formal language and axiomatic logic for reasoning
  about trust in attestations of identity (i.e. binding of a public
  key in a self-authenticating address with a traditional domain name,
  e.g., as resolvable via DNS), a logic that supports two specific
  types of trust delegation, also introduce a Kripke
  semantics for the language and show soundness of the logic with
  respect to that semantics. [Section~\ref{sec:normal-modal-logic}]

\item contrast and compare the introduced logic with logics of
  authentication, logics of local names, and authorization and access control
  logics. [Section~\ref{sec:our-logic-versus-others}]

\item prove that any statement a principal trusts is either initially
  trusted, or that there is a trust chain of delegations to this from
  initial trust assumptions. [Section~\ref{sec:ax-proof}]

\item present a saturation algorithm that effectively derives all
  possible trust statements from a finite set of initial trust
  assumptions and show it to be sound, complete, and terminating.
  [Section~\ref{sec:saturation}]

\end{itemize}

\subsection{Related Work}

Basic concepts and implementations of SATAs and sattestations as
described in this paper are set out in papers by Syverson together
with various coauthors~\cite{once-and-future,secdev19,satas-wpes21}.
\cite{secdev19} introduces a concept of contextual trust illustrated
in the ``news'' example cited above, as well as an implementation in
the form of a list of sattestations a site can post that a browser
WebExtension can select to
trust. Appendix~\ref{sec:appendix-worked-example} presents examples
following the JSON format for sattestation headers 
introduced in~\cite{satas-wpes21}, which include contextual labels.
\cite{satas-wpes21} also mentions in one sentence the possibility of
transitivity if a sattestor label is permitted, but does not say anything
more about these.  None of these papers describes any means, formal or
otherwise, to reason about when a sattestation should be trusted. Nor
do they mention any structure to labels and how these relate to each
other (e.g., as introduced below in Section~\ref{ssec:labels}). Further
they do not provide a language for, or any conception of, transitive
or iterated trust in any form.

The basic meaning we attach to ``trust'' is the same as given by Coker
et al.~\cite{coker-ijis11}: ``Principal $B$ trusts principal $A$ with
regard to the statement $\varphi$ if and only if, from the fact that
$A$ has said $\varphi$, $B$ infers that $\varphi$ was true at a given
time.'' \cite{coker-ijis11}~is about formalizing remote attestation
for trusted hardware as in the Trusted Platform Modules of the Trusted
Computing Group (TCG).  While our basic notion of trust is in this sense
the same, that paper also does not discuss transitivity or iteration
of trust or anything like our concept of contextual trust.

Formalizing transitive trust has previously been described in the
context of secure communication~\cite{ok-ok}. However, that paper is
primarily focused on honest messages (known to be true by the sender
when sent) and related concepts that help ground a theory of
cooperative communication in distributed systems. We do not assume
messages are honest, and most of our rules effectively set out
conditions under which it is possible to make a valid trust inference
from a message. There is also nothing analogous to our notion of
contextual trust therein. We further discuss transitive and iterated
trust in Section~\ref{sec:motivation-for-decisions}.

Ours is a logic of identity, particularly a logic of identity for web
addresses. We will present a detailed comparison to other logics in
Section~\ref{sec:our-logic-versus-others}, after we have set out our
concepts of identity and trust in Section~\ref{sec:identities}, and we
have set out our logic and language in
Section~\ref{sec:normal-modal-logic}.

%%% Local Variables: 
%%% mode: latex 
%%% TeX-master: "../paper"
%%% End:          

\section{Identities, Beliefs, and Trust}
\label{sec:identities}

\subsection{No entity without identity}

An \emph{entity} comprises an \emph{identity} and \emph{properties}.
The identity is signified by a domain name $\domain{}$ and an onion
address $\saname{}$.  In our notation an arbitrary $\sata{}$ is only
a potential entity, a \maybe{}. It may or may not signify a real
identity.
If $\domain{}$ and $\saname{}$ are bound together, then $\sata{}$ is a real identity.\footnote{By analogy to Thomistic ontology~\cite{being-and-essence}, we would say $\sata{}$ has \emph{esse}.} In this case we write $\langle D,O, \bound \rangle$ or,
more succinctly, $\bsata{}$.

Though it is often to be expected that any $\domain{}$ is bound with
at most one $\saname{}$ and vice versa, this need not be the case.
Amongst other things, this allows multiple domain names to be bound to
a single onion address (roughly as multiple domains might be hosted
at the same IP address) and allows multiple onion addresses bound
to a single domain name (as might be useful during an updating
of a domain's associated onion address).

Properties are expressed using labels, including the $\bound$ label just noted.\footnote{The only \emph{essential} property an
identity has is $\bound$. All other properties are
\emph{accidental}.} $\bsinglesata{}$ signifies that $\domain{}$ and $\saname{}$ are bound and that this entity has property $\ell$.
Having a property
may entail having a less specific property.  For example, a site with the 
property of being a U.S. Department of Justice (DOJ) site also has the more generic 
property of being a U.S. Government (USG) site.  A site for French recipes could have
separate properties pertaining to French things and pertaining to recipes.  It 
could also have a single ``French recipes'' property; this could entail a more generic 
``French'' property, or it could entail a more generic ``recipes'' property.\footnote{As noted when we formalize our treatment of labels in Sec.~\ref{ssec:labels}, we require that, if two properties are more generic than a third property, then one of the first two properties is more generic than the other.}  

Sattestations never assert a set of properties of which
some are claimed to be bound and others are not, and sattestations never assert multiple non-$\bound$ properties.  Thus, we do not need to worry about capturing a scenario in which $\bsata{}$ is known to have $\ell\neq\bound$ and we are reasoning about whether or not it has a different property $\ell'\neq\bound$ as well.
Also, as mentioned
above in Section~\ref{sec:introduction}, there are currently two
implemented SAT address formats in use~\cite{satas-wpes21}. The SATA
notation we have introduced in this section and that we use in our
logic is an abstraction that covers both of them.

\subsection{Belief, Knowledge, and Trust}

Suppose $P$ is a generic principal---either a client or a SATA---that forms 
beliefs about entities.  These beliefs are expressed in our logic using the \trusts\ operator, e.g., as $P\trusts\bsata{}$.  Our logic does not cope with incorrect beliefs; if $P\trusts\bsata{}$ holds, then $\domain{}$ is actually bound to $\saname{}$.  On the
other hand, we have no need of anything comparable to a knowledge (T)
axiom, because we never express $\bsata{}$ except when indicating $P\trusts \bsata{}$ or 
that $P$ asserts $\bsata{}$ (or $\blsata{}$, for some label $\ell$).

Besides bindings, the other thing a principal, $P$, may trust
is that a SATA has asserted something about the binding of an identity
(either with just the label $\bound$ or possibly an additional
label). We capture such assertions using the $\says$ operator, so that we have, e.g., 
$P \trusts (\upair{D}{O} \says \tpair{D'}{O'})$ (sometimes omitting the outer parentheses).  $P$'s 
trust in an assertion may follow from $P$ receiving or retrieving a 
sattestation header, as described in~\cite{satas-wpes21}, or it may follow from $P$ 
retrieving a list of entities claimed to be bound from the
appropriate page of a SATA-site, as described in~\cite{secdev19}.  As formalized below, 
the left-hand side of the $\says$ operator is always written as a \maybe{}, even if the 
name and address are known to be bound.

We will formally set out our language and logic in
the next section. 
For purposes of
comparison with other logics we note that all atomic formulae of our
language are in the following forms (and generalizations of them to be
explained when introduced below).

\begin{description}
  \item[relations between property labels:] $\ell \wknarrows \ell'$
      \item[trust in a binding:] $P\trusts \blsata{}$
      \item[trust that others asserted a binding:] $P\trusts
        (\upair{D'}{O'} \says \ltriple{D}{O}{\ell} )$
\end{description}

Note that it could be practical to support reasoning about arbitrary
principals asserting bindings. For example, a client on a
company-provided computer might download a list of labeled SATAs from
a file on a trusted connection to a company server without a SATA, or
Alice might be sent one or more labeled SATAs in an authenticated
email message from Bob. For simplicity, we will restrict the reasoning
about assertions in our current language and logic to reasoning about
assertions made by SATAs. 
We will simply make appropriate local trust assumptions about the
binding and labels of SATAs if based on non-SATA sources.  
Further, we assume that SATAs only make sattestations they
believe. Thus, if a SATA issues a sattestation, then what is asserted
will be reflected in the local trust assumptions of that SATA\@.

Note that this means trust about assertions are always in formulae
such as the above $P\trusts (\upair{D'}{O'} \says
\ltriple{D}{O}{\ell})$, and we never write $P\trusts (\tpair{D'}{O'}
\says \ltriple{D}{O}{\ell})$. Even if $P$ trusts that $D$ and $O$
are bound; the dependence on this in the axioms of our logic is
reflected in other antecedent clauses. The only trust a binding
assertion indicates is which \maybe{} asserted what binding.

The axioms of our logic reflect circumstances under which
it is valid to infer trust in a binding.  Our logic does not incorporate time or an entity gaining or losing a property (as might happen, even with $\bound$, when a domain is sold, a key rotated, etc.).  We assume antecedents of a derivation hold when the derivation is made.  If an antecedent ceases to be valid, reasoning can be restarted without that statement.  Our approach might be extended, e.g., by considering time intervals over which antecedents are valid (and, from those, computing a time interval over which a derivation's conclusion is valid).

Although we should expect in practice that a client would only accept
a SATA binding if it has checked an appropriately associated
(non-expired, non-revoked, CT logged, etc.) TLS Certificate, we do not
here represent such checks or associated reasoning. More broadly, our
purpose in this paper is not to set out a usable tool for formal
analysis of sattestation systems in practice, though we believe one
could be based on our logic.  Rather, our goal is to lay the basis for
a rigorous understanding of sattestation.

\subsection{Contextual Trust Model}

The central aspect of SATAs and sattestations that distinguish them
from other approaches for authentication (X.509, PGP) or trusted
naming (DNS, SPKI/SDSI) is that they are contextual rather than
structural. Note that even if PGP supports locally based trust, it is
still a purely structural determination. Trust derives from distance
and/or numbers of independent paths to primarily trusted nodes in a
web rather than X.509's adequate delegation from hierarchical roots,
but this is not based on any context (e.g., association with a
particular profession, company, or government agency); thus, it is still a
structural determination. Contexts are also not types, and we cannot
directly represent contextual trust in a standard type system.  The
same entity can have multiple contextual properties, and, as we shall
see, this makes a difference as to what can be inferred about
them. But entities might not have contextual labels at all other than
the generic \bound{} label (so, roughly, are not necessarily typed).

Also, while X.509 and DNS work primarily
at large scales, and local names and PGP work primarily at much
smaller scales, sattestations scale both up and down.  They might
operate at a large scale, such as the USG example above, or similarly
for large corporations, or organizations, in which case one would
expect large classes of browsers to come preconfigured by default to
trust those as sattestors. Besides making these widely available to
consumers, requiring such on enterprise issued or configured equipment
can also protect corporate or governmental internal connections by
remote workers against phishing attacks, VPN certificate hijack, and
other attacks that could leak sensitive information, such
as~\cite{safe-mil-warning}.

Sattestations will also operate fine, however, on a level of small
organizations, municipalities and groups of various kinds, possibly
without the resources or manpower to perform ongoing checks of CT log
monitors or other indicators of certificate problems or hijack. (And
those only provide post hoc evidence in any case.) Perhaps especially
at a small scale, a basis for contextual trust is in associations and
trust structures that already or naturally exist within the culture of
an organization. It is generally a higher bar to set oneself up to
fool the organizers and participants in a community that one is a
member of that community than to fool others with no vested interest
in that community per se.

Sattestations also do not require organizations to take
on the overhead and responsibility of becoming a CA, nor do those need
the active support of such structural entities, even if they utilize
and leverage CAs, CT logs, etc., for the trust provided. CAs might
also offer sattestation services, and this can provide some additional
assurance and attack resistance. But to the extent these services are
based on purely structural criteria for issuance of sattestations,
they will be bound for most domains by the trust that a DV (Domain
Validation) level check of identity can provide, even if they do build
security (in the form of public keys) into the web itself.

 \subsection{Trust in Implementation} \label{sec:trust-in-implementation}

Though our focus is on the logic of SATAs and sattestation, we briefly
describe their implementation and use in practice, especially in
contrast to TLS certificates. SATAs and sattestations have had
research implementation and deployment on both the client and server
sides~\cite{satdomains-traudt,satdomains-finkel}, which have also been
described in research
publications~\cite{secdev19,satas-wpes21}. On the client side 
the implementation uses a WebExtension for Firefox or Tor Browser that
checks headers such as shown in Sec.~\ref{sec:appendix-worked-example}, 
although
the cited publications note that the longer-term desirability is
implementation natively in
browsers.  In the Firefox deployment,
clients receive the strengthened authentication and hijack resistance
properties described herein. In The Tor Browser implementation,
clients also receive the routing and lookup protections that come from
routing over the Tor network and using onion service protocols (thus,
e.g., lookup via the onion-service Distributed Hash Table rather than via DNS).
Unlike Tor's onion addresses, SATAs are fully
backwards compatible: browsers knowing nothing of SATAs will offer
none of these additional protections. But access using them will
succeed as normal (assuming no attacks).

Because SATAs are represented in TLS certificates using subdomains of
registered domains, site owners can straightforwardly obtain needed DV
certificates for free, e.g., from Let's Encrypt. Because SATAs are
generally formatted in one of the two formats described in
\S~\ref{sec:introduction}, no reconfiguring
is necessary to make a site SATA-compatible. Besides the TLS
certificate just noted, a front end is all that is needed to
provide the sattestation headers.

%%% Local Variables: 
%%% mode: latex 
%%% TeX-master: "../paper"
%%% End:          

\section{Formal Language, Logic, and Semantics}
\label{sec:normal-modal-logic}

\subsection{Labels}\label{ssec:labels}

We start with the labels used in our logic.  We assume that we have a ``basic'' label set $\allsflabels$ whose elements are free of the
distinguished operators \satt{}\ and \sattdel{}. (These
operators will be explained below.)  We expand this to a set
$\alllabels$ of ``general'' labels by single applications of either
\sattestor{} or \sattestordel{}:
\[
\alllabels = \allsflabels \cup \bigcup_{\ell\in\allsflabels}
\{\satt{\ell}\} \cup \bigcup_{\ell\in\allsflabels} \{\sattdel{\ell}\}
\]
In presenting rules below, we assume that labels $\ell$, $\ell'$,
$\ell_i$ are elements of $\allsflabels$, and we will use, \eg,
$\sflabelset{}'$ and $\sflabelset{i}$ for subsets of $\allsflabels$.
We assume labels $g$, $g'$, and $g_i$ are elements of $\alllabels$,
and we will use, \eg, $\anylabelset{}'$ and $\anylabelset{i}$ for
subsets of $\alllabels$.

The operators \satt{}\ and \sattdel{}, which take as a single argument $\ell
\in\allsflabels$, always occur as labels for bound SATAs.  E.g.,
$\ltriple{D}{O}{\satt{\ell}}$ indicates that \sata{} may assert the
binding of $\ell$ to any SATA, while $\ltriple{D}{O}{\sattdel{\ell}}$
indicates that \sata{} may indefinitely delegate the binding of label
$\ell$. In other words \sata{} may assert for any $(D',O')$
that $\ltriple{D'}{O'}{\ell}$ or that $\ltriple{D'}{O'}{\satt{\ell}}$ or that
$\ltriple{D'}{O'}{\sattdel{\ell}}$. These notions are made more
precise by our axioms and truth conditions below and are further glossed
after those are presented.

We assume that there is a partial order $\narrows$ on $\allsflabels$; if $\ell_1 \narrows \ell_2$, then we say that $\ell_1$ ``narrows'' $\ell_2$ and that $\ell_2$ ``broadens'' $\ell_1$.  We use $\wknarrows$ to allow for $\ell_1 = \ell_2$, and we use $\broadens$ and $\wkbroadens$ for (weak) broadening.  $\allsflabels$ includes the distinguished label $\bound{}$, which we discuss further below and which is incomparable to all other labels in $\allsflabels$.  We make the following additional assumptions about $\narrows$:\footnote{While the first two assumptions are stated in terms of what principals ``can determine,'' these are assumptions about information conveyed by labels that a principal may not have seen before rather than principals' computational power.  These assumptions are satisfied by the label system we use here and the canonical label representation described in the following paragraphs.}
\begin{itemize}
    \item For every $\ell\in\allsflabels$, the set $\uparrow\{\ell\} = \{\ell' | \ell' \wkbroadens \ell\}$ is finite.  Furthermore, given $\ell$, any principal can determine the elements of $\uparrow\{\ell\}$.
    \item Given $\ell_1$ and $\ell_2$, any principal can determine whether $\ell_1 \narrows \ell_2$, $\ell_1 \broadens \ell_2$, $\ell_1 = \ell_2$, or $\ell_1$ and $\ell_2$ are incomparable under $\narrows$.
    \item For every $\ell$, $\uparrow\{\ell\}$ is linearly (totally) ordered under $\narrows$.
\end{itemize}

Equivalent to the third condition above is that $\narrows$ induces the structure of an ordered, rooted forest on $\allsflabels$.  The root of each tree in the forest is the $\narrows$-maximum element of that tree.
We observe that
\[
\ell_1 \wkbroadens \ell_2 \text{ iff } \uparrow\{\ell_1\} \subseteq \uparrow\{\ell_2\},
\]
with equality holding on one side iff it holds on the other.

If every $\ell\in\allsflabels$ can be represented by a distinct finite string, then, given our first and third assumptions on the label poset, we may without loss of generality represent each $\ell\in\allsflabels$ as a finite string $\lambda(\ell)$, where these strings satisfy the property that $\ell_1 \wkbroadens \ell_2$ if, and only if, $\lambda(\ell_1)$, concatenated with a distinguished symbol, is a prefix of $\lambda(\ell_2)$.

To see the previous claim: Let `$\mathsf{:}$' be a distinguished symbol, and (re-encoding if necessary) let $\sigma(\ell)$ be a finite, `$\mathsf{:}$'-free string representing the label $\ell$ (and no $\ell'\neq \ell$).  Given $\ell$, we define its string label $\lambda(\ell)$ as follows.  Let $\widehat{\ell}$ be the maximum element in $\uparrow\{\ell\}$.  Define $\lambda(\widehat{\ell})$ as $\sigma(\ell)$.  Inductively, let $\ell'$ be the largest element of $\uparrow\{\ell\}$ for which $\lambda(\ell')$ has not yet been defined.  Let $\ell''$ be the label covering $\ell'$.\footnote{$x$ covers $y$ iff $x\broadens y$ and $x\wkbroadens z\wkbroadens y$ implies $z = x$ or $z = y$.}  We define $\lambda(\ell')$ as the concatenation (indicated by the $\widehat{\ }$ operator) $\lambda(\ell'')\ \widehat{\ }\ \text{`}\mathsf{:}\text{'}\ \widehat{\ }\ \sigma(\ell')$.

Because $\lambda(\ell')$ is defined inductively in terms of the string representations of the elements of the finite, linearly ordered set $\uparrow\{\ell'\}$, it will not be defined in conflicting ways in the process of defining the $\lambda$ labels of multiple elements of $\allsflabels$.

By our inductive construction, in this canonical representation of labels, $\lambda(\ell_1)\ \widehat{\ }\ \text{`}\mathsf{:}\text{'}$ is a prefix of $\lambda(\ell_2)$ if, and only if, $\ell_1\wkbroadens \ell_2$.  Furthermore, if $s = \sigma(\ell)$ and $\lambda$ is any canonical label representation, then $\text{`}\mathsf{:}\text{'}\ \widehat{\ }s\ \widehat{\ }\ \text{`}\mathsf{:}\text{'}$ appears at most once as a substring of $\text{`}\mathsf{:}\text{'}\ \widehat{\ }\lambda\ \widehat{\ }\ \text{`}\mathsf{:}\text{'}$.

Our examples of labels will all have the form of canonical $\lambda$-labels as just described and as specified in the following grammar.
\begin{eqnarray*}
\sigma &=& \text{`}\mathsf{:}\text{'-free string that is not `}\mathsf{bound}\text{'}\\
\lambda &=& \text{`}\mathsf{bound}\text{'}\ |\ \sigma\ |\ \lambda\ \widehat{\ }\ \text{`}\mathsf{:}\text{'}\ \widehat{\ }\ \sigma\\
& &\ \text{ (where }\text{`}\mathsf{:}\text{'}\ \widehat{\ }\sigma\ \widehat{\ }\ \text{`}\mathsf{:}\text{'}\text{ is not a substring of }\text{`}\mathsf{:}\text{'}\ \widehat{\ }\lambda\ \widehat{\ }\ \text{`}\mathsf{:}\text{')}
\end{eqnarray*}
If we use two labels $\lambda_1, \lambda_2$, then we have $\lambda_1\broadens\lambda_2$ if, and only if, $\lambda_1\ \widehat{\ }\ \text{`}\mathsf{:}\text{'}$ is a prefix of $\lambda_2$.

\begin{example}
Assume `$\mathsf{USG}$', `$\mathsf{DOJ}$', and `$\mathsf{FBI}$' are $\sigma$-strings 
corresponding to the U.S. Government, Department of Justice, and the Federal Bureau of 
Investigation, respectively.  Then:
\begin{eqnarray*}
\mathsf{USG} &\broadens& \mathsf{USG:DOJ:FBI},\\
\mathsf{recipes} &\broadens& \mathsf{recipes:French},\\
\mathsf{USG:DOJ:FBI} &\broadens& \mathsf{USG:DOJ:FBI:recipes:French}
\end{eqnarray*}
but the label `$\mathsf{USG:DOJ:FBI:recipes}$' and the label
`$\mathsf{recipes:French}$' are incomparable. 
\end{example}
Because we are using the canonical label format described above, no string may appear multiple times as a $\sigma$-string in a label; e.g., 
`$\mathsf{recipes:USG:DOJ:FBI:recipes}$' is not permitted as one of our example labels.

\subsection{Grammar, axioms, and rules}

We now present a grammar (Sec.~\ref{sssec:grammar}) for terms that may appear in our logic, our logical axioms and local axiom schemata  (Sec.~\ref{sssec:axioms}), and the deduction rules (Sec.~\ref{sssec:rules}) that we use.  The local axiom schemata describe the forms of axioms that a principal may optionally use; they capture trust decisions specific to a principal and not inherent in our logic.  We provide further discussion of the axioms and local axiom schemata in Sec.~\ref{sec:motivation-for-decisions}.

\subsubsection{Grammar}\label{sssec:grammar}
We start with a grammar for the terms that may appear in our logic.
\begin{eqnarray*}
\mathsf{maybe} &:=& \upair{\mathsf{domain}}{\mathsf{onion}};\\
\mathsf{princ} &:=& P\in\mathcal{P}\ |\ \tpair{\mathsf{domain}}{\mathsf{onion}} |\ \mathsf{maybe};\\
\mathsf{bd\_label} &:=& ``\mathsf{bound}";\\
\mathsf{bd\_pair} &:=& \tpair{\mathsf{domain}}{\mathsf{onion}}\ |\ \ltriple{\mathsf{domain}}{\mathsf{onion}}{\mathsf{bd\_label}};\\
\mathsf{sf\_label} &:=& \ell \in \allsflabels;\\
\mathsf{gen\_label} &:=& g \in \alllabels;\\
\mathsf{bd\_l\_pair} &:=& 
\mathsf{bd\_pair}\ |\ \ltriple{\mathsf{domain}}{\mathsf{onion}}{\mathsf{gen\_label}};\\
\mathsf{s\_stmt} &:=& \mathsf{maybe} \says \mathsf{bd\_l\_pair};\\
\mathsf{t\_stmt} &:=& \mathsf{princ} \trusts \mathsf{s\_stmt}\ |\ \mathsf{princ} \trusts \mathsf{bd\_l\_pair};\\
\mathsf{wk\_order} &:=& \mathsf{sf\_label} \wknarrows \mathsf{sf\_label};\\
\mathsf{stmt} &:=& \mathsf{bd\_pair}\ |\ \mathsf{s\_stmt}\ |\ \mathsf{t\_stmt}\ |\ \mathsf{wk\_order};\\
\mathsf{form} &:=& \mathsf{stmt}\ |\ \mathsf{form} \mcond \mathsf{form} \ |\ \mathsf{form} \land \mathsf{form};
\end{eqnarray*}

The definition of $\mathsf{s\_stmt}$s captures that, as noted above, only 
$\mathsf{maybe}$s make assertions via $\says$ statements.

\subsubsection{Axioms and local axiom schemata}\label{sssec:axioms}

We start with our axioms.

\begin{spacing}{.7}

\begin{enumerate}[label=A\arabic*]

\item\label{ax:1wt}
\begin{equation*}
\begin{split}
    &\Bigl(\ell_1 \wkbroadens \ell_2\ \land\\
    &\bigl(P\trusts \ltriple{D}{O}{\ell_2}\bigr)\Bigr)\\
    &\mcond\ P\trusts \ltriple{D}{O}{\ell_1}
\end{split}
\end{equation*}

\item\label{ax:2wt}
\begin{equation*}
\begin{split}
    &\Bigl(\ell_1 \wkbroadens \ell_2 \wkbroadens \ell_3\ \land\\
    &P\trusts \ltriple{D_1}{O_1}{\satt{\ell_1}}\ \land \\
    &\bigl(P\trusts \upair{D_1}{O_1} \says \ltriple{D_2}{O_2}{\ell_2}\bigr) \land\\
    &\bigl(P\trusts \upair{D_2}{O_2} \says \ltriple{D_2}{O_2}{\ell_3}\bigr)\Bigr)\\
    &\mcond P\trusts \ltriple{D_2}{O_2}{\ell_2}
\end{split}
\end{equation*}

\item\label{ax:3wt}
\begin{equation*}
\begin{split}
    &\Bigl(\ell_1 \wkbroadens \ell_2 \wkbroadens \ell_3\ \land\\
    &P\trusts \ltriple{D_1}{O_1}{\sattdel{\ell_1}}\ \land \\
    &\bigl(P\trusts \upair{D_1}{O_1} \says \ltriple{D_2}{O_2}{\sattdel{\ell_2}}\bigr)\ \land\\
    &\bigl(P\trusts \upair{D_2}{O_2} \says \ltriple{D_3}{O_3}{\satt{\ell_3}}\bigr)\Bigr)\\
    &\mcond P\trusts \ltriple{D_2}{O_2}{\sattdel{\ell_3}}
\end{split}
\end{equation*}

\item\label{ax:4wt}
\begin{equation*}
\begin{split}
    &\Bigl(\ell_1 \wkbroadens \ell_2 \wkbroadens \ell_3\ \land\\
    &P\trusts \ltriple{D_1}{O_1}{\sattdel{\ell_1}} \land \\
    &\bigl(P\trusts \upair{D_1}{O_1} \says \ltriple{D_2}{O_2}{\satt{\ell_2}}\bigr) \land\\
    &\bigl(P\trusts \upair{D_2}{O_2} \says \ltriple{D_3}{O_3}{\ell_3}\bigr)\Bigr)\\
    &\mcond P\trusts \ltriple{D_2}{O_2}{\satt{\ell_3}}
\end{split}
\end{equation*}

\item\label{ax:5wt}
\begin{equation*}
\begin{split}
    &P\trusts \ltriple{D}{O}{g} \\ 
    &\mcond \ P\trusts \ltriple{D}{O}{\bound}
\end{split}
\end{equation*}

\item\label{ax:6wt}
\begin{equation*}
\begin{split}
    &\bigl(P\trusts \upair{D}{O}\says \varphi\bigr)\\
    &\mcond \bigl(P\trusts \upair{D}{O}\says\ltriple{D}{O}{\bound}\bigr)
\end{split}
\end{equation*}

\item\label{ax:7wt} 
\begin{equation*}
\begin{split}
    &P\trusts \ltriple{D}{O}{\sattdel{\ell}}\\
    &\mcond P\trusts \ltriple{D}{O}{\satt{\ell}}
\end{split}
\end{equation*}

\item\label{ax:8wt} 
\begin{equation*}
\begin{split}
    &P\trusts \upair{D_1}{O_1}\says \ltriple{D_2}{O_2}{\sattdel{\ell}}\\
    &\mcond P\trusts \upair{D_1}{O_1}\says \ltriple{D_2}{O_2}{\satt{\ell}}
\end{split}
\end{equation*}

\item\label{ax:9wt} 
\begin{equation*}
\begin{split}
    &\ell_1 \wkbroadens \ell_2\ \land\\
    &\bigl(P\trusts \ltriple{D}{O}{\satt{\ell_1}}\bigr)\\
    &\mcond\ \bigl(P\trusts \ltriple{D}{O}{\satt{\ell_2}}\bigr))
\end{split}
\end{equation*}

\item\label{ax:10wt} 
\begin{equation*}
\begin{split}
    &\ell_1 \wkbroadens \ell_2\ \land\\
    &\bigl(P\trusts \ltriple{D}{O}{\sattdel{\ell_1}}\bigr)\\ 
    &\mcond\ \bigl(P\trusts \ltriple{D}{O}{\sattdel{\ell_2}}\bigr)
\end{split}
\end{equation*}

\item\label{ax:11wt}
\begin{equation*}
\begin{split}
    &\upair{D}{O}\says \varphi\\
    &\mcond \bigl(P \trusts \upair{D}{O}\says \varphi\bigr)
\end{split}
\end{equation*}

\item\label{ax:12wt}
    \begin{enumerate}
    \item\label{ax:12awt}
       \begin{equation*}
        (\varphi \land \psi) \mcond \varphi
      \end{equation*}
    \item\label{ax:12bwt}
      \begin{equation*}
        (\varphi \land \psi) \mcond (\psi \land \varphi)
      \end{equation*}
    \item\label{ax:12cwt}
      \begin{equation*}
        (\varphi \land (\psi \land \delta)) \mcond ((\varphi \land \psi)
        \land \delta)
      \end{equation*}
    \end{enumerate}

\end{enumerate}

\end{spacing}

We also specify the form of local axiom schemata that a principal may, optionally, 
use to reflect local trust.  These are parameterized either by two $\mathsf{gen\_label}$s 
($\lasdel{}{}$) or by a $\mathsf{domain}$ and $\mathsf{onion}$ ($\lassatt{}{}$ and 
$\lassattdel{}{}$).

\begin{spacing}{.7}

\begin{multline}
P \trusts \langle D,O,g_1 \rangle  \ \mcond \ P \trusts \langle D,O,g_2 \rangle\\ \text{for given }g_1, g_2\text{ and for any }D, O.\tag{$\lasdel{g_1}{g_2}$}\label{las:del}
\end{multline}

\begin{multline}
P \trusts \langle D,O, \satt{\ell} \rangle\text{ for given }D, O\\
\text{and for any and all labels }\ell.\tag{$\lassatt{D}{O}$}\label{las:satt}
\end{multline}

\begin{multline}
P \trusts \langle D,O, \sattdel{\ell} \rangle\text{ for given }D, O\\
\text{and for any and all labels }\ell.\tag{$\lassattdel{D}{O}$}\label{las:sattdel}
\end{multline}

\end{spacing}

\subsubsection{Rules}\label{sssec:rules}

Finally, we note the logical rules that we allow.

\begin{equation}
\varphi\text{ and }\varphi \mcond \psi \ \implies \ \psi \tag{$\mathsf{MP}$}\label{rule:mp}
\end{equation}
\begin{equation}
\varphi\text{ and } \psi \ \implies \ \varphi \land \psi \tag{$\mathsf{\land I}$}\label{rule:ai}
\end{equation}

We now describe the reasoning reflected in our axioms and motivation
for these as well as for our linguistic and logical choices, after
which we will present our semantics.

\subsection{Motivation for linguistic and logical design decisions}
\label{sec:motivation-for-decisions}

This section describes the intent behind our axioms and rule schemata
as well as decisions for the linguistic and logic design choices we
have made.

The only rules are Modus Ponens ($\mathsf{MP}$) and Adjunction ($\mathsf{\land I}$, a.k.a.\ And-Introduction). Local axiom schemata
reflect local trust assumptions particular principals might hold
about nonlogical trust relations between labels or which principals to
trust as sattestor.

Axiom~\ref{ax:1wt} effectively says that any SATA that $P$ trusts to
have a particular property label is also trusted to have any broader
label, e.g., being a U.S. Government Dept.\ of Justice SATA (label $\mathsf{USG:DOJ}$) 
implies being a U.S. Government SATA (label $\mathsf{USG}$).
\ref{ax:1wt} guides other design choices and how labels should
be interpreted. It is important to understand that in our view labels do
not capture authorization.  We instead view narrowing a label as
providing \emph{more specific} information: for $\ell\broadens\ell'$,
we view $\ltriple{D}{O}{\ell'}$ as more specific than
$\ltriple{D}{O}{\ell}$.

The \bound{} label is incomparable to all other labels. If it were narrower
than one or more other labels, \ref{ax:1wt} would allow us to infer that a SATA
had those labels simply because it is bound. And if \bound{}
were broader than one or more other labels, \ref{ax:9wt} would allow a SATA that
is trusted to have $\satt{\bound{}}$ to sattest
to \ltriple{D}{O}{\ell'} for any $D$ and $O$ and any of those narrower labels $\ell'$.  
An analogous problem would also arise with~\ref{ax:10wt} and $\sattdel{\bound{}}$.

Axioms~\ref{ax:2wt}--\ref{ax:4wt} tell us a label that one can infer is
bound to a particular SATA based on a label that a trusted sattestor
has attested is bound to that SATA\@. The first
antecedent conjunct of each of these (about the relation between the
three labels occuring in the other antecedent conjuncts) allows for
the possibility that the labels in each of them
need not be identical but can be appropriate narrowings or broadenings
of the others.  These axioms also reflect (in the fourth antecedent
conjunct of each) the goal of authority independence.  As originally
stated, authority independence means that a misbehaving or misled CA
will not be trusted by itself to bind an onion address and a domain in
a TLS certificate. The owner of the domain and the onion address must
concur for this binding to be trusted by a SATA-aware
client~\cite{secdev19,satas-wpes21}. The authority independence aspect
introduced herein applies not to CAs but to sattestors: a sattestor's
assertion that a sattestee has a label will not be trusted unless the
sattestee concurs.  For Axiom~\ref{ax:2wt}, this is an explicit
assertion of a label $\ell_3$ that weakly narrows the sattested label;
for Axioms~\ref{ax:3wt} and~\ref{ax:4wt}, this is a statement in which
the sattestee acts as a sattestor in the appropriate way.  We need
three axioms to cover the different kinds of labels in the ultimate
consequent: sattestor-free, delegatable sattestor, or simple
sattestor, respectively.

For practical simplicity,
a sattestor knows the label \satlabel{i} for which it is providing a
sattestation and only states it under a role qualified to do so.  We
assume any sattestor can detect and resist an attacker attempting
to somehow confuse that sattestor about the state it is in so as to
cause it to utter a sattestation under a sattestor role different from
one appropriate for the sattestation uttered. This allows us to avoid
having to include a label \satt{\satlabel{j}} (or respectively
\sattdel{\satlabel{j}}) under which \sata{1} is making a sattestation
in the second trust conjunct of Axiom~\ref{ax:2wt}---or indeed having
a sattestor making any statement be anything beyond a $\maybe{}$.

We do not include an analogue of Axiom~\ref{ax:2wt} that replaces the
restriction $\ell_1 \wkbroadens \ell_2 \wkbroadens \ell_3$ with the
restrictions $\ell_1 \wkbroadens\ell_2$ and $\ell_3 \broadens\ell_2$,
\ie, in which the label claimed by the sattestee is (strictly) broader
than the label given by the sattestor. We now explain why.

If we derived the narrower label given by the sattestor---\ie, we
derived $P\trusts \ltriple{D_2}{O_2}{\ell_2}$---then we would be
giving the sattestee a label that is more specific than what it claims
about itself.  This contradicts the motivation for requiring the
sattestee to say a label consistent with what the sattestor says.  For
example, we would not want a sattestor to be able to give the label
$\mathsf{USG:DOJ:FBI:recipes:French}$ to a site that only claims the
label $\mathsf{USG:DOJ:FBI}$ for itself.

If we derived the broader label claimed by the sattestee---\ie, we
derived $P\trusts \ltriple{D_2}{O_2}{\ell_3}$---then the concern is a
bit more subtle.  As an example, consider labels that reflect the
structure of a hierarchical organization.  Suppose that $(D_1, O_1)$
is a directory site in the engineering division of ACME Corp. that can
label other sites as research sites within that division, \ie, it has
the property $\satt{\mathsf{ACME:Eng:research}}$.  Suppose that
$(D_2,O_2)$ is an internal server for the marketing division of ACME
Corp.\ $(D_2,O_2)$ claims that it is an ACME Corp.\ site by saying it
has the label $\mathsf{ACME}$.  If $(D_1,O_1)$ says that $(D_2,O_2)$
has the label $\mathsf{ACME:Eng:research}$, then---using the
hypothetical axiom we are considering---we could derive that
$(D_2,O_2)$ has the label $\mathsf{ACME}$.  On one hand, this seems
reasonable; if $(D_2,O_2)$ had the label
$\mathsf{ACME:Marketing:internal}$, then it would also have the label
$\mathsf{ACME}$ (by~\ref{ax:1wt}).  On the other
hand, the derivation that we are considering makes use of a
sattestation by a site $(D_1,O_1)$ that is in a different
administrative division from $(D_2,O_2)$ and whose \satt{}\ authority
does not cover the division containing $(D_2,O_2)$.  On balance, we
think this is undesirable, so we do not include such a derivation
here.  If $(D_2,O_2)$ wanted the broader label $\ell_3$---here,
$\mathsf{ACME}$---based on the sattestation from $(D_1,O_1)$, then it
needs to claim the narrower label $\ell_2$---here
$\mathsf{ACME:Eng:research}$---so that
$\ell_1\wkbroadens\ell_2\wkbroadens\ell_3$, and~\ref{ax:2wt} would
apply. 

Many notions called `trust' are transitive, though this is also
sometimes criticised~\cite{trust-not-transitive}. Following previous
work on SATAs, our notion is not inherently transitive. Just because
$P$ trusts that a $\ltriple{D}{O}{\ell}$ is bound and has property
$\ell$, does not mean that $P$ trusts sattestations by $(D,O)$.  And
even if $P$ trusts sattestations by $(D,O)$, it may not trust $(D,O)$
to sattest who else is a sattestor.  See~\cite{secdev19} for an
illustrative example.

Rather than simply prohibit transitive trust, our axioms (above) and truth
conditions (below) set out the contextual assumptions under which we support
it.

Axiom~\ref{ax:5wt} ensures that a SATA trusted to be bound with any label is
also trusted to be simply bound.

Axiom~\ref{ax:6wt} tells us that if a principal trusts that a SATA
has asserted anything, that SATA is trusted to be implicitly
asserting that it is itself a bound SATA\@.

Axiom~\ref{ax:7wt} indicates that when a principal trusts that a SATA is
a delegatable-sattestor of a label it also trusts that SATA is
simply a sattestor of that same label. Axiom~\ref{ax:8wt} reflects
that principal trust in an assertion about a delegatable-sattestor
label similarly implies trust in an assertion about a simple sattestor
of the label. Note that
\ref{ax:8wt}, combined with \ref{ax:7wt}, allows for a
simplification of the logic with respect to
Axioms~\ref{ax:2wt}--\ref{ax:4wt}.  It could have been necessary to have
two versions of each of those axioms: wherever $\satt{}$ occurs in an
antecedent clause of those axioms we could have also needed to have a
version of the axiom with $\sattdel{}$ in that clause. But
Axioms~\ref{ax:7wt} and~\ref{ax:8wt} allow us to infer trust in
a simple sattestor label (assertion about such a label) from
trust in a delegatable-sattestor label (assertion about such a label).

Axioms~\ref{ax:9wt} and~\ref{ax:10wt} describe trust that
a (delegatable) sattestor of any label is also a (delegatable)
sattestor of any narrowing of that label.

Axiom~\ref{ax:11wt} indicates that any statement is effectively a public
announcement. Every principal knows every assertion made by any
SATA\@. (Principals do not necessarily trust what is asserted; they
only trust that it is asserted.) Note that we are not concerned in
this logic with capturing reasoning about the sending of messages or
authenticated propagation of information. Rather we are concerned with
what can be inferred from assertions made by any principal together
with assumptions regarding who is trusted to assert which bindings.

Axiom~\ref{ax:12wt} comprises collectively the standard axioms for use
of the sentential connective for conjunction. We will in general be
liberal in our writing of conjunctions when the usage is clear and
consistent with propositional logic, thus allowing, e.g., `$\varphi
\land \psi \land \delta$'. (We have similarly allowed, .e.g., `$\ell_1
\wkbroadens \ell_2 \wkbroadens \ell_3$'.)

For the local axiom schemata, a principal would use an instance of \ref{las:del} to trust that any SATA bound to a label $g_1$ is also bound to the label $g_2$.  The \ref{las:satt} and \ref{las:sattdel} schemata allow a principal to trust a specified SATA to be a (delegatable) sattestor for any label.  This is useful if the principal wants to trust whatever her employer (or a trustworthy friend, etc.) says.

\subsection{Semantics}

A \emph{frame} $\mathcal{F}$ comprises a tuple $(W,\{R_{P}\})$
where $W$ is a set of worlds, and for each principal $P$ there is a relation 
$R_P \subseteq W \times W$ on worlds such that $(w,w') \in R_P$ indicates
$w'$ is accessible from $w$ by $P$.  A model $\mathcal{M}$ consists of: a frame $\mathcal{F}$, 
a client $c$, a set $\mathcal{D}$ of domain names, a set $\mathcal{O}$ of onion addresses, a label set $\allsflabels$ and a partial 
order $\narrows$ on it, the set $\Phi$ of all well-formed formulae (all $\mathsf{form}$s in our grammar above), and a (world-dependent) assignment 
function $a_w$
based on those.  Note that all names, labels (and the partial order on them), and statements are
shared globally across worlds, while local differences (between worlds) are all
reflected in the assignment function $a_w$. So, in a frame $\mathcal{F} = (W,\{R_{P}\})$, a world $w \in W$ is a
tuple $(c, \mathcal{D}, \mathcal{O}, \allsflabels, \narrows, \Phi, a_w)$.
Because we do not attempt in this logic to capture faulty trust
reasoning (where a principal can trust something that is false), our
accessibility relation on worlds $R$ is reflexive.

Ours is not a normal modal logic, primarily because we do not capture
trust of arbitrary formulae.  Given our motivating application, our focus is ultimately on 
reasoning about contextual trust (the binding of labeled SATAs) while relying on the 
types of utterances that are made in sattestations.  While extensions of our logic may be 
of future interest from a logical perspective, here we intentionally restrict our 
language to trust in SATA bindings and specified types of assertions. 
Thus, as already noted, in our current language only a
$\maybe{}$ can say anything.  Further, the only thing that can be said
is a bound SATA, which may have any type of label. And the only
formulae that a principal can trust are those expressing such bound
SATAs and those expressing utterances of bound SATAs by a $\maybe{}$
SATA\@.  To minimize notation we do not introduce separate notation
for the formulae that a $\maybe{}$ can say or for the formulae that a
principal can trust. But it should be kept in mind that these cannot
be arbitrary formulae.

To capture binding of identities we introduce a pair of functions
$\beta_1,\beta_2$ on $\mathcal{D} \times \mathcal{O}$, where
$\beta_1(\sata{}) =$ the set of all
$\sata{i} \in \mathcal{D} \times \mathcal{O}$ s.t. $\domain{i}$ is
bound with $\saname{}$, and \\
$\beta_2(\sata{}) =$ the set of all
$\sata{i} \in \mathcal{D} \times \mathcal{O}$ s.t. $\saname{i}$ is
bound with $\domain{}$. Those preliminaries out of the way, we now
set out the conditions that our assignment function must satisfy.  The specification of an assignment function $a_w$ satisfying these then determines the truth conditions.  For a world $w$ and a label $\ell$, $a_w(\ell)$ defines those $(D,O)$ for which $\ltriple{D}{O}{\ell}$ in $w$. Similarly, $a_w(\Phi_{(D,O)})$ defines the set of statements that $(D,O)$ utters in $w$.  We discuss constraints on the assignment function in Sec.~\ref{sec:motivation-for-decisions}.
\begin{enumerate}[label=AF\arabic*]
  \item $a_{w}(c) = c$, $a_{w}(\domain{} \in \mathcal{D}) =
    \domain{}$, $a_{w}(\saname{} \in \mathcal{O}) = \saname{}$
  \item $a_w(\ell\in\widehat{\Lambda}) \subseteq \mathcal{D} \times \mathcal{O}$ 
  \item $a_w(\Phi_{(D,O)}) \subseteq \{ \varphi | \varphi \text{ is a }\mathsf{bd\_l\_pair} \}$
  \item $a_w(\beta_1(\sata{})) \subseteq \beta_1(\sata{})$,\\
    $a_w(\beta_2(\sata{})) \subseteq \beta_2(\sata{})$,  
\end{enumerate}

By uttering a formula, any SATA implicitly
asserts that it is an entity. We thus restrict $a_w(\Phi_{\sata{}})$ so
that if $a_w(\Phi_{\sata{}})$ is nonempty, then
$\langle D,O, \bound \rangle \in a_w(\Phi_{\sata{}})$.  (In general, note that an element of $a_w(\Phi_{(D,O)})$ is a $\mathsf{bd\_l\_pair}$ $\pi$ and not 
an $\mathsf{s\_stmt}$ of the form $(D,O)\says \pi$.)
Also, the set of SATAs with a given label can vary from world to
world. Nonetheless, we do not permit that $\ell' \narrows \ell$ in one
world but not in another. Thus if $a_w(\ell') \subseteq a_w(\ell)$ at
any world $w$, then $a_{w'}(\ell') \subseteq a_{w'}(\ell)$ at all
worlds $w'$. We reflect this in our truth conditions, which we set
out following some preliminary remarks.

We give truth conditions to expressions in our language in terms of
satisfaction at a world $w$. (We take the model, $\mathcal{W}$ to be
set and implicit.)  We begin with the various logically atomic
formulae: first expressing the narrowing relation on labels, then
various types of formulae for SATAs bound with different types of
labels. Next we present truth conditions for formulae reflecting
assertions, then trust formulae, and finally formulae containing
logical connectives.

Our semantics is model-theoretic in structure, not
operational, but we are also motivated by an operational
perspective reflected in the truth conditions for
various individual kinds of formulae.  Together, the truth conditions allow us to evaluate $w\models \phi$, where $\phi$ is a $\mathsf{form}$ in our grammar.
Truth conditions rely on the function $a_w$ and truth conditions weakly earlier in the list; for convenience, the first line of each truth condition indicates the type of $\phi$ for which it defines the truth of $w\models\phi$.
\begin{enumerate}[label=TC\arabic*]
  \item {[$\phi$ is $\mathsf{stmt:wk\_order}$]}\\ 
  $w \models \ell \wknarrows \ell'$ iff \\
   \big($\ell'\ \widehat{}:$ is an initial string of $\ell$\big) $\land$\\
    \big( for all $w'$, $a_{w'}(\ell) \subseteq a_{w'}(\ell')$\big)
    \label{tc:wkorder}

  \item {[$\phi$ is $\mathsf{stmt:bd\_pair}$]}\\ 
  $w \models \langle D, O, \bound \rangle$ iff \\
    $w \models \langle D, O \rangle$ iff \\
    \big($a_w(D) = \mathsf{proj}_1(D',O')$ \\
    for some $(D',O') \in a_w(\beta_1(D,O))$\big) $\land$ \\
    \big($a_w(O) = \mathsf{proj}_2(D'',O'')$ \\
    for some $(D'',O'') \in a_w(\beta_2(D,O))$\big)
    \label{tc:bdpair}

  \item {[$\phi$ is $\mathsf{stmt:t\_stmt:bd\_l\_pair}$ with $g\in\allsflabels\setminus\{\bound\}$]}\\ 
  $w \models \blsata{}$ iff \\
    \big($w \models \langle D, O, \bound \rangle$\big) $\land$ \big($(D,O) \in a_w(\ell)\big)$
    \label{tc:bdlpair}

  \item 
  {[$\phi$ is $\mathsf{stmt:t\_stmt:bd\_l\_pair}$ with $g=\satt{\ell}$]}\\ 
  $w \models \langle D, O, \satt{\ell} \rangle$ iff \\
    \big($w \models \langle D, O, \bound \rangle$\big) $\land$ \\
    \Big(if \big($\langle D',O',\ell' \rangle \in a_w(\Phi_{\sata{}})$ $\land$ \\
    \ \ \ $\langle D',O',\ell'' \rangle \in a_w(\Phi_{(D',O')})$ $\land$\\
    \ \ \ $w\models \ell'' \wknarrows \ell'$ $\land$\\
     \ \ \ $w\models \ell' \wknarrows \ell$\big)\\ 
   then 
   $w\models \langle D',O',\ell' \rangle$\Big)
    \label{tc:satt}

  \item 
  {[$\phi$ is $\mathsf{stmt:t\_stmt:bd\_l\_pair}$ with $g=\sattdel{\ell}$]}\\ 
  $w \models \langle D_1, O_1, \sattdel{(\ell)} \rangle$ iff \\
    $w\models \langle D_1, O_1, \satt{(\ell)} \rangle$ $\land$\\
    \Bigg(
    if \big($w\models \ell'' \wknarrows \ell'$ $\land$ $w\models \ell' \wknarrows \ell$\big), then\\
    \bigg(\Big(if \big($\langle D_2,O_2, \sattdel{\ell'} \rangle \in a_w(\Phi_{\sata{1}})$ $\land$\\ $\langle D_3,O_3, \satt{\ell''} \rangle \in a_w(\Phi_{\sata{2}})$\big) then \\ $w\models \langle D_2,O_2, \sattdel{\ell''} \rangle$\Big) $\land$\\
     \Big(if \big($\langle D_2,O_2, \satt{\ell'} \rangle \in a_w(\Phi_{\sata{1}})$ $\land$\\
     $\langle D_3,O_3, \ell'' \rangle \in a_w(\Phi_{\sata{2}})$ \big) then\\
     $w\models \langle D_2,O_2, \satt{\ell''} \rangle$ \Big)\bigg)
    \Bigg)
    \label{tc:sattdel}

  \item {[$\phi$ is $\mathsf{stmt:s\_stmt}$, i.e., of form $\mathsf{maybe}\says\mathsf{bd\_l\_pair}$]}\\ 
  $w \models \sata{} \says \varphi$ iff \\
    \big(for all $P$, $R_P(w,w') \Rightarrow \varphi \in a_{w'}(\Phi_{\sata{}})$\big) $\land$\\
    \big(if $\varphi = \langle D', O', \sattdel{\ell} \rangle$ for some $D'$, $O'$, and $\ell$,\\ then
    $\langle D', O', \satt{\ell} \rangle \in a_{w'}(\Phi_{\sata{}})$\big)
    \label{tc:sstmt}

  \item {[$\phi$ is $\mathsf{stmt:t\_stmt}$; $\varphi$ must be $\mathsf{s\_stmt}$ or $\mathsf{bd\_l\_pair}$]}\\ 
  $w \models P \trusts \varphi$ iff \\
    $\forall w'\in W$ \big(if $R_P(w,w')$, then  $w'\models \varphi$\big)
    \label{tc:trust}

  \item {[$\phi$ is $\mathsf{form}\mcond\mathsf{form}$]}\\ $w \models \varphi \mcond \psi $ iff\\
  $w\not\models \varphi$ or $w\models\psi$.
\label{tc:mcond}
    
  \item {[$\phi$ is $\mathsf{form}\land\mathsf{form}$]}\\ $w \models \varphi \land \psi $ iff \\
  $w\models \varphi$ and $w\models \psi$.
    \label{tc:land}
\end{enumerate}

Note that it is possible for circular dependencies involving~\ref{tc:sattdel}.  Thus, e.g., the
condition $w\models \ltriple{D}{O}{\sattdel{\ell}}$ might depend
on $w \models \ltriple{D'}{O'}{\sattdel{\ell}}$ and vice versa.
For our motivating scenario in
Section~\ref{sec:appendix-worked-example}, this is not a problem: as
shown by Thm.~\ref{thm:sattdel}, a principal will not \emph{derive}
trust in a SATA of the form $\ltriple{D}{O}{\sattdel{\ell}}$ unless
this is rooted in an explicit local assumption (and not a cycle of
dependencies). We now further describe scenarios motivating 
choices for the semantics of $\sattdel{}$.

\arxivonly{
\paragraph{Repeated delegation of same label for redundancy} A (hypothetical)
motivating case for the $\sattdel{}$ operator is: the Department of
Defense (DoD) is a trusted root for saying that a server is usable by
operational forces via the label $\mathsf{DoDops}$.  These forces may
be operating in a communication-denied/-degraded environment.  If one
server $\ltriple{D_i}{O_i}{\mathsf{DoDops}}$ is being replaced or
augmented by another server $(D_j, O_j)$, it may not be possible to
get DoD to provide a sattestation for $(D_j, O_j)$ in a timely manner.
It would be useful if $(D_i, O_i)$ could do this, so that there would
be a chain of trust from DoD to
$\ltriple{D_i}{O_i}{\sattdel{\mathsf{DoDops}}}$.  $(D_i, O_i)$ could
then provide a sattestation of
$\ltriple{D_j}{O_j}{\sattdel{\mathsf{DoDops}}}$; if needed, this would
allow $(D_j, O_j)$ to provide a sattestation of
$\ltriple{D_k}{O_k}{\mathsf{DoDops}}$, \etc
} %end arxivonly

\paragraph{Delegating narrower authority without needing to know how it
  might be narrowed further} A (hypothetical) motivating case for the
\sattdel{}\ operator that makes use of label narrowing is: the General
Services Administration (GSA, assumed to be a trusted root for
$\mathsf{USG}$) sattests to
$\ltriple{D_1}{O_1}{\sattdel{\mathsf{USG:DOJ}}}$ for the DOJ.
\sata{1} sattests to
$\ltriple{D_2}{O_2}{\sattdel{\mathsf{USG:DOJ:FBI}}}$ for the FBI\@.
This eventually works its way down to a sattestation that effectively
asserts that the FBI's New York field office website is a
U.S. Government website.  GSA should not need to know DOJ or FBI
structure and delegations in order for this to work.

\paragraph{Delegation is permitted to ``break'' narrowing}
The the U.S. Dept.\ of Justice, Bureau of Justice Assistance maintains
on the web~\cite{doj-privacy-orgs} a list of civil liberties
organizations including the Electronic Frontier Foundation (EFF). One
could imagine using SATAs to capture (with stronger authentication)
trust in this list by using local axiom schema~\ref{las:del} to state
that if $P$ trusts \ltriple{D}{O}{\mathsf{USG:DOJ}} then $P$ trusts
\ltriple{D}{O}{\sattdel{\mathsf{civil\_liberties}}}. This could then
permit a trust that is ultimately rooted in a GSA SATA with the label
$\sattdel{\mathsf{USG}}$ to permit trust in an EFF SATA with label
$\mathsf{civil\_liberties}$ even though, $\mathsf{civil\_liberties}
\not\wknarrows \mathsf{USG}$.
Section~\ref{sec:appendix-worked-example} presents a worked example
elaborating on this and illustrating the implementation of it.

\subsection{Soundness and incompleteness}

\subsubsection{Soundness}

A \emph{derivation} is a sequence of formulae, where
each line in the derivation is either an assumption or follows
from the set of assumptions $\Gamma$ by an application of one of
the two rules to previously derived lines or instances of an axiom.
Assumptions include instances of local trust axiom schemata and
specific trust statements that a principal has preloaded, each
being either a bound SATA or an assertion by a SATA about a bound SATA.
Assumptions also include facts about narrowings and broadenings of
labels otherwise occurring in the set of assumptions.

We present new notation needed to state our soundness result.  Let
$\Gamma$ stand for a finite set of formulae\footnote{We recognize the
notation overload from the use of `$\anylabelset{}$' for a set of
labels introduced above. We nonetheless stick with the common use of
`$\Gamma$' for a finite set of formulae when speaking of derivations and
soundness, and we rely on context to keep the usages clear.}  with and
$\varphi$ as before be an arbitrary formula.  `$\Gamma \vdash
\varphi$' means that $\varphi$ is derivable from $\Gamma$ and the
axioms using the inference rules of the logic.  `$\Gamma \semcons
\varphi$' means that, in all models, $\varphi$ is true at all the
worlds at which all the members of $\Gamma$ are true.
Our use of `$\semcons$' follows~\cite{conventionbook}.
The more common notation corresponding to this use of `$\semcons$' is
`$\models$'. But, `$\models$' is being used to represent satisfaction
of a formula at a given world. We thus avoid this notational overload.

\begin{theorem}{\bf (Soundness)}
If $\Gamma \vdash \varphi$, then $\Gamma \semcons \varphi$.
\end{theorem}

\begin{proof}
To begin we need the following lemma.

\begin{lemma}
All axioms are valid in all models.
\end{lemma}

\begin{proof}
This result follows directly by inspection of the truth conditions,
which amounts to tedious case checking. We thus set out
just one axiom as an example. We choose Axiom~\ref{ax:3wt} because it
uses all of the linguistic expressions of our language, thus invokes
multiple truth conditions.

Suppose $w$ satisfies the antecedent of Axiom~\ref{ax:3wt}, where the
labels are $\ell_1, \ell_2, \ell_3$ and the SATAs are $\sata{1}$,
$\sata{2}$, and $\sata{3}$. Then by \ref{tc:land}, $w$ satisfies
each of the conjuncts. Thus we suppose $w \models \ell_1 \wkbroadens
\ell_2 \wkbroadens \ell_3$.  Let $P$ be a principal for which $w
\models P \trusts \langle D_1, O_1, \sattdel{\ell_1} \rangle$. 
Then by \ref{tc:trust}, for all $w'$ s.t.\ $R_P(w,w')$, we have $w'
\models \langle D_1, O_1, \sattdel{\ell_1} \rangle$. Thus, by
\ref{tc:trust} applied also to the other antecedent conjuncts and then
by clause (a) of \ref{tc:sattdel}, $w' \models \langle D_2, O_2,
\sattdel{\ell_3} \rangle$. Thus by a final invocation of
\ref{tc:trust}, $w \models P \trusts \langle D_2, O_2,
\sattdel{\ell_3} \rangle$.

\end{proof}

We now proceed to prove the theorem by showing that all the ways that
$\varphi$ can follow from $\Gamma$ are ways that preserve truth. This
is also straightforward.
\begin{description}

\item[$\varphi$ is an axiom or member of $\Gamma$.] Then
$\Gamma \semcons \varphi$ trivially.

\item[$\varphi$ is obtained by modus ponens from $\psi$ and $\psi
  \mcond \varphi$.] This follows by a trivial argument via
  strong induction and by the definition of the truth conditions.
  Suppose that soundness holds for all lines of a derivation up to the
  one in question, where $\varphi$ occurs. Then, by inductive
  hypothesis, $\Gamma \semcons \psi$ and $\Gamma \semcons \psi \mcond
  \varphi$.  So, $\Gamma \semcons \varphi$ by~\ref{tc:mcond}.

  \item[$\varphi \land \psi$ is obtained by $\land$-introduction from
    $\varphi$ and $\psi$.]  The same as for MP: Suppose soundness
    holds for all lines of a derivation up to the one in question,
    where $\varphi \land \psi$ occurs. Then, by inductive hypothesis,
    $\Gamma \semcons \varphi$ and $\Gamma \semcons \psi$.  So, $\Gamma
    \semcons \varphi \land \psi$ by the definition of $a_w(\varphi
    \land \psi)$.
  
\end{description}
\end{proof}

\subsubsection{Incompleteness}

Our logic is intentionally not complete by design for practical
considerations.  As one example, our axioms are all based on premises
it is legitimate to assume and on what can follow from them. We thus
do not have to cope with negation in our logic or language. 
So if $(D,O) \not\in \beta_1((D,O)) \cup \beta_2((D,O))$,
then at all $w$, $w \not\models \tpair{D}{O}$.
Given any $\ell \neq \ell'$, such that $\ell \narrows \ell'$ or
$\ell'\narrows \ell$, by \ref{tc:mcond}, at all $w$ both
$w \models \tpair{D}{O} \mcond \ell \wknarrows \ell'$ and
$w \models \tpair{D}{O} \mcond \ell' \wknarrows \ell$.
But both formulae cannot be derivable in our logic (indeed, neither is).

\section{Our logic of identities versus other logics}
\label{sec:our-logic-versus-others}

We present a more detailed comparison below, but in this paragraph we 
list specific distinguishing features of our logic making it not a simple
extension (nor restriction) of others. Authentication logics have no
means to reason about delegation or contextual trust. Logics for local
names support delegation but not global entities with local
properties.  Delegation logics have also been proposed for
authorization and access control. Sattestation is an assigned
permission, so somewhat similar.  But, unlike authorization, we place
no restrictions on which principals can create a sattestation. There
are however conditions for when to trust what is sattested. Also
unlike either authentication logics or authorization logics, our logic
supports authority independence: for a principal to trust that a
$\maybe{}$ is an entity or an entity with particular properties, the
principal must trust that the entity has itself concurred.
Finally, logics of authentication and
authorization generally address the distribution of session keys to
appropriate principals or propagation of authorizations to appropriate
relying parties.  The logic we present herein does not capture any
sending or receiving of messages or acquiring of new
information and is incompatible with those logics.
Indeed, as reflected in our Axiom~\ref{ax:11wt}, it is not
possible for any principal to say anything unless all principals know
that she did.

Burrows, Abadi, and Needham in their logic of
authentication~\cite{ban-tocs90} introduced a construct that
formalizes basic trust as used by us and, as noted above, by Coker et
al. In the language of BAN logic (as it is commonly known) `$P$ {\bf
  controls} $\varphi$' means that $P$ has jurisdiction over the truth
of $\varphi$.  They also have constructs for binding keys to
principals, which allows them to have ``message-meaning'' rules in
their logic, that say, e.g., if $P$ believes $K$ is $Q$'s public key,
and $P$ sees a message stating $\varphi$ that is signed with the
private cognate of $K$, then $P$ believes $Q$ said $\varphi$. A
primary sort of statement made in messages that principals (e.g., key
servers) send is about a session key being good for two principals to
communicate.  Thus if the conclusion of an application of a message
meaning rule is \mbox{$A$ {\bf believes} $S$ {\bf said} {\em
    good}($K,\{A,B\}$)} if also $A$ believes that $S$ controls
statements about the goodness of session keys, then their
``jurisdiction'' rules allows the derivation of \mbox{$A$ {\bf
    believes} {\em good}($K,\{A,B\}$)}.  (The jurisdiction rule
actually refers to what $A$ believes $S$ believes rather than what $S$
says, but since BAN is also a logic for honest principals---who only
say things that they believe---we elide this as aside for our
purposes.) Note also, however, that this is still reflecting $A$
privately receiving $K$ about which she is forming a belief. Indeed,
$A$ only believes {\em good}($K,\{A,B\}$) if she believes that
awareness of $K$ is limited to principals $A$ and $B$ (and $S$ or
other trusted server).

Our logic thus shares much with BAN logic. Though they do
analyze a draft X.509 protocol from the eighties, their work is all
pre-web and is thus not designed for reasoning about TLS (or SSL)
authentication. More importantly, they explicitly recognize the
complexity of delegation of jurisdiction and note 
that they do not attempt to reason about nested jurisdiction~\cite{ban-tocs90}.
Further, though they can express jurisdiction over particular
statements, they have no notion comparable to our contextual labels
much less one with structure such as in our narrowing relation.
BAN logic was reformulated and given a Kripke semantics by Abadi and
Tuttle~\cite{at91}. In this respect, our logic is closer to the
Abadi-Tuttle reformulation than to the original BAN logic. Although,
as previously noted, our logic of trust is not a normal modal logic
while their logic of belief is. In this sense we may be closer to the
original BAN, which limited beliefs to things like key bindings,
utterances made in encrypted messages, and jurisdiction, rather than
representing beliefs about formulae containing logical connectives.

Note that our labels are meant to convey descriptive information about
SATAs. They are simply not roles that can be authorized to perform an
action. Trust in a sattestation might depend on the labels that
the sattestation asserts apply to the sattestee, as well as labels
already trusted to apply to the sattestor. But the only action SATAs
can be ``authorized'' (more accurately `trusted') to do is sattestation.

Various delegation logics have been published that capture forms of
transitive or iterated trust, and that do express contextual names and
contextual trust, for example as formalizations of reasoning about
SPKI (Simple Public Key Infrastructure), SDSI (Simple Distributed
Security Infrastructure), or their combination.
(See~\cite{logical-reconstruction-of-spki} and related work referenced
therein for papers on reasoning about delegation and local names.)

Our contextual labels are not local names (nor global names). Our
labels are not names at all but descriptive labels for principals,
more akin to predicates than individuals to which predicates apply.
It would, nonetheless, be straightforward to express linked-name-based
contextual trust in our logic if reflected in SATAs.

Also, logics for
local names generally have linking axioms giving principals complete
jurisdiction over their local namespace. But sattestor jurisdiction
over a contextual label is not absolute. For example, even if a SATA
is trusted as sattestor for a label, we cannot infer that a sattestee
has that label unless the sattestee has also asserted of itself that
it has that label, or a narrowing of it.
(See Axiom~\ref{ax:2wt} and associated discussion in 
Sec.~\ref{sec:motivation-for-decisions}.) Conversely, a principal will also not automatically trust the label in
a self-sattestation.

Leaving aside for the moment labels (besides $\bound$), SATAs do
provide a sort of local-name. These are, however, comprised of the
binding of two different types of global names. Unlike local names, it
is only the binding that is local, and a (local) sattestor cannot use
locally either global name (domain or onion address) that someone else
controls as part of its local namespace. And like other labels, trust
that $\langle D, O, \bound \rangle$ depends on either trust that
$\sata{}$ has $\sattestor{}$ label wrt itself, or trust in another
sattestor for $\langle D, O, \bound \rangle$ \emph{and} that $\sata{}$
has implicitly or explicitly asserting such binding about itself.

While we do have something that is very roughly akin to authorization
(sattestation), this is primarily in support of inferring trust that a
$\maybe{}$ is an entity or is an entity with a particular property
label.  Access control and authorization logics primarily support
inferences about whether a principal (or in the case of trust
management, a key) should have read or write access to a
resource. Reasoning about the binding of an entity to a key
(authentication) is generally not part of an authorization logic per
se. Also, unlike authorizations, we place no restriction on who may
perform a sattestation. And our logics are designed to work in
environments where everyone has global access to everything (except
private keys).

%%% Local Variables: 
%%% mode: latex 
%%% TeX-master: "../paper"
%%% End:          

\section{Worked Example}
\label{sec:appendix-worked-example}

In this section, we describe example scenarios wherein a SATA-aware
browser receives sattestations that are JSON formatted as in the
implementation described in~\cite{satas-wpes21}. And, as a result, the
browser is then able to infer trust in a SATA with specific labels
based on initial trust assumptions.

Assume principal $P$ is a browser with a WebExtension capable of
processing sattestations, evaluating whether it has sufficient
sattestations to trust various SATAs. We will assume for simplicity
that when contacting a SATA, $P$ receives from that SATA, or already
possesses, all sattestations needed for such judgments. We leave as
beyond our current scope the important questions of credential
discovery and delivery for various practical scenarios as well as any
discussion of how and if trust negotiation should be managed. Also out
of scope is any discussion of how $P$ stores or processes
assumptions or inferences; when we state what $P$ trusts, we intend
by this only what $P$ assumes or is entitled to infer based on our
logic.

Assume that $P$ already trusts\\
$\ltriple{\texttt{gsa.gov}}{\mathsf{onion1}}{\sattdel{\mathsf{USG}}}$,
i.e., that a SATA for U.S. General Services Administration (GSA) is
delegatable sattestor for U.S. Government SATAs. Assume further that
$P$ receives the JSON formatted sattestations in
Figures~\ref{fig:gsa-sattestation},\ref{fig:doj-sattestation}.
(In these figures and throughout this appendix, actual onion
addresses are replaced with $\mathsf{onion}n$ in interest of space.
Similarly convenience applies to sattestor signatures in figures.)

\begin{figure}
\begin{center}
\small
\begin{verbatim}
{ "sattestation":  {
    "sattestation_version":1,
    "sattestor_domain":"gsa.gov",
    "sattestor_onion":"onion1",
    "sattestor_refresh_rate":"7 days",
    "sattestees": [
    {
    // bind domain to a self auth. address
      "domain": "justice.gov",
      "onion": "onion2",   // onion address
      "labels": "USG:DOJ",
                "sattestor*(USG:DOJ)",
      "issued": "2023-12-08",
      "refreshed_on": "2024-01-25"
    },
    {
    // bind domain to a self auth. adress
      "domain": "hhs.gov",
      "onion": "onion3",   // onion address
      "labels": "USG:HHS",
                "sattestor*(USG:HHS)",
      "issued": "2023-12-08",
      "refreshed_on": "2023-01-25"
    } ... ]  },
  // signature by sattestor
  "signature": "sig-gsa-1" }
\end{verbatim}
\caption{\small A sattestation by GSA of SATAs for various USG Departments}
\label{fig:gsa-sattestation}
\end{center}
\end{figure}

\begin{figure}
\begin{center}
\small
\begin{verbatim}
{ "sattestation":  {
    "sattestation_version":1,
    "sattestor_domain":"justice.gov",
    "sattestor_onion":"onion2",
    "sattestor_refresh_rate":"7 days",
    "sattestees": [
    {
    // bind domain to a self auth. address
      "domain": "fbi.gov",
      "onion": "onion4",   // onion address
      "labels": "USG:DOJ:FBI",
                "sattestor*(USG:DOJ:FBI)",
      "issued": "2023-12-14",
      "refreshed_on": "2024-01-27"
    },
    {
    // bind domain to a self auth. address
      "domain": "bja.ojp.gov",
      "onion": "onion5",   // onion address
      "labels": "USG:DOJ:OJP:BJA",
      "issued": "2023-12-14",
      "refreshed_on": "2024-01-27"
    }  ]  },
  // signature by sattestor
  "signature": "sig-justice-1" }
\end{verbatim}
\caption{\small A sattestation by DOJ of FBI and BJA SATAs}
\label{fig:doj-sattestation}
\end{center}
\end{figure}

\begin{figure}
\begin{center}
\small
\begin{verbatim}
{ "sattestation":  {
    "sattestation_version":1,
    "sattestor_domain":"fbi.gov",
    "sattestor_onion":"onion4",
    "sattestor_refresh_rate":"7 days",
    "sattestees": [
    {
    // bind domain to a self auth. address
      "domain": "fbi.gov",
      "onion": "onion4",   // onion address
      "labels": "USG:DOJ:FBI",
                "sattestor*(USG:DOJ:FBI)",
      "issued": "2023-12-20",
      "refreshed_on": "2024-01-20"
    }  ]  },
  // signature by sattestor
  "signature": "sig-fbi-1" }
\end{verbatim}
\caption{\small A self-sattestation by an FBI SATA}
\label{fig:fbi-self-sattestation}
\end{center}
\end{figure}

Suppose Alice is looking at an FBI webpage. Her browser, $P$, is
configured with the trust assumptions we have described and has
received sattestations as in the
Figures~\ref{fig:gsa-sattestation}--\ref{fig:fbi-self-sattestation}.
We will use our axioms to derive from this
$P \trusts \ltriple{\mathsf{fbi.gov}}{\mathsf{onion4}}{\mathsf{USG:DOJ:FBI}}$.
We begin by using an instance of Axiom~\ref{ax:3wt} to
derive\\ $P \trusts
\ltriple{\mathsf{doj.gov}}{\mathsf{onion2}}{\sattdel{\mathsf{USG:DOJ}}}$.

$\mathsf{USG} \wkbroadens \mathsf{USG:DOJ} \wkbroadens \mathsf{USG:DOJ:FBI}$, which gives us the first conjunct of the axiom antecedent.
Our initial trust assumption provides the second conjunct.
From receipt of the Figure~\ref{fig:gsa-sattestation} sattestation, we
get\\
$P \trusts \upair{\mathsf{gsa.gov}}{\mathsf{onion1}}
\says$\\
\hspace*{\fill}$\ltriple{\mathsf{doj.gov}}{\mathsf{onion2}}{\sattdel{\mathsf{USG:DOJ}}}$,\\
providing the third conjunct.  From receipt of the
Figure~\ref{fig:doj-sattestation} sattestation we also know that\\
$P \trusts\upair{\mathsf{doj.gov}}{\mathsf{onion2}}
\says$\\
\hspace*{\fill}$\ltriple{\mathsf{fbi.gov}}{\mathsf{onion4}}{\sattdel{\mathsf{USG:DOJ:FBI}}}$.\\ Thus
by Axiom~\ref{ax:8wt}, we also have the fourth conjunct, $P
\trusts\upair{\mathsf{doj.gov}}{\mathsf{onion2}}
\says$\\
\hspace*{\fill}$\ltriple{\mathsf{fbi.gov}}{\mathsf{onion4}}{\satt{\mathsf{USG:DOJ:FBI}}}$.\\
We now have all the antecedents of Axiom~\ref{ax:3wt}, allowing us to apply the axiom and derive its consequent.

If $P$ has also received a self-sattestation from an FBI SATA as in
Figure~\ref{fig:fbi-self-sattestation}, we have
$P \trusts\upair{\mathsf{fbi.gov}}{\mathsf{onion4}} \says$
$\ltriple{\mathsf{fbi.gov}}{\mathsf{onion4}}{\mathsf{USG:DOJ:FBI}}$.
Then by an application of Axiom~\ref{ax:2wt} similar to the
application of Axiom~\ref{ax:3wt} just illustrated, we can derive\\
$P \trusts \ltriple{\mathsf{fbi.gov}}{\mathsf{onion4}}{\mathsf{USG:DOJ:FBI}}$.

Continuing to build on the above, our next example illustrates both the
applicability of local axioms and how delegation can effectively
``break'' narrowing. Suppose that $P$ also trusts Justice Department
sattestations concerning which SATAs are civil liberties
organizations and which SATAs are competent to issue sattestations thereof.

This is reflected in instances of local axiom
schema\\
\ref{las:del}, in particular $\lasdel{\mathsf{USG:DOJ}}{\mathsf{civil\_liberties}}$\\
and $\lasdel{\mathsf{USG:DOJ}}{\sattdel{\mathsf{civil\_liberties}}}$,
as noted in the above discussion leading into this worked example.

\vspace*{-.3\baselineskip}
\begin{equation*}
  \begin{split}
    &A \trusts \ltriple{D}{O}{\mathsf{USG:DOJ}}\\
    &\mcond A \trusts \ltriple{D}{O}{\mathsf{civil\_liberties}}\\
    & \\
    &A \trusts \ltriple{D}{O}{\mathsf{USG:DOJ}}\\
    &\mcond A \trusts \ltriple{D}{O}{\sattdel{\mathsf{civil\_liberties}}}
  \end{split}
\end{equation*}
\vspace*{-.7\baselineskip}

At time of writing, the Dept.\ of Justice, Bureau of Justice
Assistance (BJA) maintains on the web a list of privacy and civil
liberties organizations~\cite{doj-privacy-orgs}. To provide
authenticated, hijack resistant credentials for both the addresses of
these organizations and endorsement of them being competent privacy
and civil liberties organizations, a BJA SATA could provide
sattestations as in Figure~\ref{fig:bja-sattestation}.

  \begin{figure}
\begin{center}
\small
\begin{verbatim}
{ "sattestation":  {
    "sattestation_version":1,
    "sattestor_domain":"bja.ojp.gov",
    "sattestor_onion":"onion5",
    "sattestor_refresh_rate":"7 days",
    "sattestees": [
    {
    // bind domain to a self auth. address
      "domain": "aclu.org",
      "onion": "onion6",   // onion address
      "labels": "civil_liberties",
                "sattestor*(civil_liberties)",
      "issued": "2023-12-05",
      "refreshed_on": "2024-01-24"
    },
    {
    // bind domain to a self auth. address
      "domain": "adl.org",
      "onion": "onion7",   // onion address
      "labels": "civil_liberties",
                "sattestor*(civil_liberties)",
      "issued": "2023-12-05",
      "refreshed_on": "2024-01-24"
    },
    {
    // bind domain to a self auth. address
      "domain": "eff.org",
      "onion": "onion8",   // onion address
      "labels": "civil_liberties",
                "sattestor*(civil_liberties)",
      "issued": "2023-12-05",
      "refreshed_on": "2024-01-24"
    },
    {
    // bind domain to a self auth. address
      "domain": "epic.org",
      "onion": "onion9",   // onion address
      "labels": "civil_liberties",
                "sattestor*(civil_liberties)",
      "issued": "2023-12-05",
      "refreshed_on": "2024-01-24"
    },
    {
    // bind domain to a self auth. address
      "domain": "splc.org",
      "onion": "onion10",   // onion address
      "labels": "civil_liberties",
                "sattestor*(civil_liberties)",
      "issued": "2023-12-05",
      "refreshed_on": "2024-01-24"
    }  ]  },
  // signature by sattestor
  "signature": "sig-bja-1" }
\end{verbatim}
\caption{\small Sattestation by the DoJ Bureau of Justice Assistance of civil liberties SATAs}
\label{fig:bja-sattestation}
\end{center}
\end{figure}

Note that unlike the prior figures, Figure~\ref{fig:bja-sattestation}
does not include sattestee labels reflecting a narrowing of labels for
which the sattestor is assumed or derived to be trusted by $P$. Nonetheless,
because of the local axioms mentioned above, it can ground derivation of
$P \trusts \ltriple{\texttt{eff.org}}{\mathsf{onion8}}{\mathsf{civil\_liberties}}$
in $P$'s trust of
$\ltriple{\texttt{gsa.gov}}{\mathsf{onion1}}{\sattdel{\mathsf{USG}}}$.

And despite this break of narrowing, because that same sattestation by
the Bureau of Justice Assistance SATA covers
$\ltriple{\texttt{eff.org}}{\mathsf{onion8}}{\sattdel{\mathsf{civil\_liberties}}}$,\\
this could continue.
For example, if $P$ receives a (current)
sattestation signed by this EFF SATA for\\
$\ltriple{\texttt{torproject.org}}{\mathsf{onion11}}{\mathsf{civil\_liberties:tech\_partner}}$
as well as a current self-sattestation from that
Tor Project SATA, then this is sufficient to ground\\
$P \trusts$\\
$\ltriple{\texttt{torproject.org}}{\mathsf{onion11}}{\mathsf{civil\_liberties:tech\_partner}}$.

\section{Properties of Trust Derivations}
\label{sec:ax-proof}

We prove properties of the trust derivations obtained with our logic.  These properties describe conditions that must hold if $P$ is able to derive trust in a $\mathsf{bd\_l\_pair}$ $\ltriple{D}{O}{g}$.  We give separate results depending on the form of $g$.  These cover whether $g$ is: \bound{} (Prop.~\ref{prop:bound}); a \satt{}-free label $\ell\in\allsflabels\setminus\{\bound{}\}$ (Prop.~\ref{prop:sf}); of the form \satt{\ell} for some $\ell\in\allsflabels$ (Prop.~\ref{prop:satt}); or of the form \sattdel{\ell} for some $\ell\in\allsflabels$ (Thm.~\ref{thm:sattdel}).  Here, we present the proof of Prop.~\ref{prop:satt}.  The other proofs, presented in the extended version~\cite{jsm24arxiv} of this paper, are similar.

Each of these results guarantees that one of a list of conditions must hold, including the possibility that a term of a certain form appears in \itrust{P} ($P$'s initial trust assumptions).  Other possibilities include that certain types of terms are derivable or that \las{P} ($P$'s trusted instances of local-axiom schemata) includes a certain schema instance.  In the case of Thm.~\ref{thm:sattdel}, our most complex result, there is a chain of triples $(D_i,O_i,\ell_i)$, ending with $(D,O,\ell)$, in which successive triples show how $P$'s trust is transferred and whose initial triple shows that this trust is rooted in something local to $P$ (either a local trust assumption or a local axiom schemata).

If the conclusion holds that certain terms are derivable, at least one such term is such that we may then again apply one of our results.  This allows us to eventually show that the derivation is rooted in something explicitly trusted by $P$ (as an element of \itrust{P} or \las{P}) and that (if applicable) certain sattestations were made by a chain of SATAs and that $P$ trusted corresponding SATAs as sattestors.  We state this formally as a combined result that gives all the possible conditions when starting with the derivability of $P\trusts\ltriple{D}{O}{g}$ for an arbitrary $g$; however, we think that the lengthy case enumeration that would entail would not provide any more insight or utility than what we state here.

\subsection{Definitions and assumptions}
\label{sec:trust-derivation-defs}

We use \itrust{P} to capture the trust that $P$ assumes rather than derives.  We assume that the elements of \itrust{P} are all $\mathsf{t\_stmt}$s of the form $P\trusts X$.  Considering our grammar, each of these is conjunction free.

We use \las{P} for the set of local-axiom schema instances that $P$ assumes.  Each element of \las{P} is assumed to be of one of the forms given above (\ie, \lasdel{g_1}{g_2}, \lassatt{D}{O}, and \lassattdel{D}{O}).  We note that antecedent and conclusion of an instantiation of \lasdel{g_1}{g_2}, and the terms in instantiations of \lassatt{D}{O}, and \lassattdel{D}{O}, are all conjunction free.

We use \tstmts\ for the set of \says\ statements that $P$ sees and may use for derivations.   Each element of \tstmts\ is an $\mathsf{s\_stmt}$ and thus of the form $\upair{D}{O}\says\varphi$, where $\varphi$ is a $\mathsf{bd\_l\_pair}$.  Recall that, using \ref{rule:mp} and Axiom~\ref{ax:11wt}, if $\upair{D}{O}\says\varphi\in\tstmts$, then we can derive $P\trusts \upair{D}{O}\says\varphi$.  Considering our grammar, each element of \tstmts\ is conjunction free.

We consider derivations of terms that use applications of \ref{rule:mp} and \ref{rule:ai} to the axioms and the elements of \itrust{P}, \las{P}, and \tstmts.  As observed in the definitions of these sets, each element of these sets is conjunction free.  Thus, if a derivation of a term $\varphi$ includes the application of \ref{rule:mp} to Axiom~\ref{ax:12awt} (to extract one conjunct from a conjunction), there is also a derivation of $\varphi$ that does not include any applications of \ref{rule:mp} to Axiom~\ref{ax:12awt}.  We thus assume that any derivation in which we are interested does not include any applications of \ref{rule:mp} to Axiom~\ref{ax:12awt}.

Recall that $\bound{}$ is incomparable with every other label.

\subsection{Formal results}

The proofs of these results are similar to each other.  Each fixes a derivation $\Delta$ of a statement $P\trusts\ltriple{D}{O}{g}$ and iterates backwards through the derivation.  Inspection of the various axioms and local axiom schemata identifies which instantiations of these could produce the statement we are then considering.  If $g = \bound{}$ (Prop.~\ref{prop:bound}), then the trust statement may be in \itrust{P}; otherwise, we immediately have that $P\trusts\ltriple{D}{O}{g'}$, for $g'\neq\bound{}$, or that another SATA sattests to $\ltriple{D}{O}{\bound{}}$.  If $g\neq\bound{}$ is sattestor-free (Prop.~\ref{prop:sf}), then moving backwards through the derivation narrows the label associated with \tpair{D}{O} until reaching one that is sattested to by another SATA, is produced by an instance of \lasdel{}{}, or is in \itrust{P}.  If $g = \satt{\ell}$ (Prop.~\ref{prop:satt}), which might arise after working backwards from the sattestor-free case, then moving backwards through the derivation can broaden the label in the argument to $\satt{}$; this broader label may serve as an argument to a \satt{} label bound to \tpair{D}{O} (in \itrust{P}, produced by \lassatt{}{}, or produced by \lasdel{}{}).  A broader label may also serve as an argument to \sattdel{}, either bound to \tpair{D}{O} or bound to a different SATA \tpair{D'}{O'} that sattests to a \satt{} label for \tpair{D}{O}.  Finally, if $g = \sattdel{\ell}$ (Thm.~\ref{thm:sattdel}), then there are conditions for initial trust and local schema analogous to the $\satt{\ell}$ case.  The difference here is that there can be a sequence of multiple sattestations in which each SATA sattests to a \sattdel{} property of the next SATA in the sequence.

We now turn to the formal statements of our results that we have just sketched.  The proof of Prop.~\ref{prop:satt} is presented here; the other proofs are in the extended version~\cite{jsm24arxiv} of this paper.

\begin{proposition}\label{prop:bound}
If
\begin{equation*}
P \trusts \ltriple{D}{O}{\bound{}}
\end{equation*}
is derivable by repeatedly applying Rules~\ref{rule:mp} and~\ref{rule:ai} to the axioms, the instances of the Local Axiom Schemata in \las{P}, the elements of \tstmts, and the elements of \itrust{P}, then one of the following conditions holds:
\begin{enumerate}
    \item For some $g\in\alllabels$, $P \trusts \ltriple{D}{O}{g}$ is derivable;\label{bd:dog} or
    \item For some $\upair{D'}{O'}$:\label{bd:satt}
    \begin{enumerate}
        \item $P \trusts \ltriple{D'}{O'}{\satt{\bound{}}}$ is derivable; and
        \item $P \trusts \upair{D'}{O'} \says \ltriple{D}{O}{\bound{}}$ is derivable;
    \end{enumerate}
    or
    \item $P \trusts \ltriple{D}{O}{\bound{}}$ is in $\itrust{P}$.
\end{enumerate}
\end{proposition}

\arxivonly{
\begin{proof}[Proof of Prop.~\ref{prop:bound}]
Let $\gamma = P \trusts \ltriple{D}{O}{\bound{}}$ and consider a derivation $\Delta$ of $\gamma$.

\begin{enumerate}
    \item If $\gamma$ is produced by an application of \ref{rule:mp} to Axiom~\ref{ax:1wt}, then (because \bound{} is incomparable to all other labels), the antecedent in this instance of the axiom was also $\gamma$, and we consider the previous step in $\Delta$.

    \item If $\gamma$ was produced by an application of \ref{rule:mp} to Axiom~\ref{ax:2wt}, then this instance of the axiom has, as a conjunct in its antecedent, the term $P \trusts \ltriple{D'}{O'}{\satt{\bound{}}}$ for some \upair{D'}{O'}.  The antecedent also includes a conjunct $P \trusts \upair{D'}{O'} \says \ltriple{D}{O}{\bound{}}$ (for the same \upair{D'}{O'}).  Thus, we satisfy the claimed Condition~\ref{bd:satt}.
    
    \item If $\gamma$ was produced by an application of \ref{rule:mp} to Axiom~\ref{ax:5wt}, then the term $P \trusts \ltriple{D}{O}{g}$ is derivable for some $g\in\alllabels$.  Similarly, if $\gamma$ was produced by an application of \ref{rule:mp} to an instance of $\lasdel{g}{\bound{}}\in\las{P}$, then $P \trusts \ltriple{D}{O}{g}$ is again derivable.  In either of these cases, we satisfy the claimed Condition~\ref{bd:dog}.
\end{enumerate}

No other steps could produce $\gamma$ in $\Delta$ if it was not in $\itrust{P}$, completing the proof.
\end{proof}
} % end of \arxivonly

\begin{proposition}\label{prop:sf}
If, for $\ell\in\allsflabels$, $\ell\neq\bound{}$,
\begin{equation*}
P \trusts \ltriple{D}{O}{\ell}
\end{equation*}
is derivable by repeatedly applying Rules~\ref{rule:mp} and~\ref{rule:ai} to the axioms, the instances of the Local Axiom Schemata in \las{P}, the elements of \tstmts, and the elements of \itrust{P}, then one of the following conditions holds:
\begin{enumerate}
    \item $P \trusts \ltriple{D}{O}{\widetilde{\ell}}$ is in $\itrust{P}$ for some $\widetilde{\ell}\wknarrows\ell$;\label{sf:itrust} or
    \item For some $g\in\alllabels$ and $\widetilde{\ell}\wknarrows\ell$:\label{sf:del}
    \begin{enumerate}
        \item $\lasdel{g}{\widetilde{\ell}}\in\las{P}$; and
        \item $P\trusts \ltriple{D}{O}{g}$ is derivable;
    \end{enumerate}
    or
    \item For some $\upair{D'}{O'}$ and $\widehat{\ell}$ and $\widetilde{\ell}$ such that $\widehat{\ell}\wkbroadens\widetilde{\ell}\wknarrows\ell$:\label{sf:satt}
    \begin{enumerate}
        \item $P \trusts \ltriple{D'}{O'}{\satt{\widehat{\ell}}}$ is derivable; and
        \item $P \trusts \upair{D'}{O'} \says \ltriple{D}{O}{\widetilde{\ell}}$ is derivable
    \end{enumerate}
\end{enumerate}
\end{proposition}

\arxivonly{
\begin{proof}[Proof of Prop.~\ref{prop:sf}]
Fix a derivation $\Delta$ of $\gamma = P \trusts \ltriple{D}{O}{\ell}$.  Start with $\widetilde{\gamma} = \gamma$.  Inductively, we assume $\widetilde{\gamma}$ has been defined to be of the form $P\trusts\ltriple{D}{O}{\widetilde{\ell}}$ for $\widetilde{\ell}\in\allsflabels$ and $\widetilde{\ell}\wknarrows\ell$, and we consider how it was produced in $\Delta$.

\begin{enumerate}
    \item If $\widetilde{\gamma} \in \itrust{P}$, then we have satisfied Condition~\ref{sf:itrust} and we terminate the induction.
    
    \item If $\widetilde{\gamma}$ results from an application of \ref{rule:mp} to Axiom~\ref{ax:1wt}, then the antecedent of the axiom instance includes as a conjunct $P\trusts\ltriple{D}{O}{\ell'}$ for some $\ell'\wknarrows\widetilde{\ell}$.  We update $\widetilde{\gamma}$ to be this conjunct (setting $\widetilde{\ell} = \ell'$) and continue the induction.
    
    \item If $\widetilde{\gamma}$ results from an application of \ref{rule:mp} to Axiom~\ref{ax:2wt}, then the antecedent of this instance of the axiom includes as conjuncts, for some $\widehat{\ell}\wkbroadens\widetilde{\ell}$ and some \upair{D'}{O'}, the terms $P \trusts \ltriple{D'}{O'}{\satt{\widehat{\ell}}}$ and $P \trusts \upair{D'}{O'} \says \ltriple{D}{O}{\widetilde{\ell}}$.  This satisfies the claimed Condition~\ref{sf:satt}, and we terminate the induction.
    
    \item If $\widetilde{\gamma}$ results from the application of \ref{rule:mp} to an instance of \ref{las:del}, then that instance must be of the form $\lasdel{g}{\widetilde{\ell}}$ for some $g\in\alllabels$, and $P\trusts\ltriple{D}{O}{g}$ is derivable.  This satisfies the claimed Condition~\ref{sf:del}, and we terminate the induction.
\end{enumerate}
By inspection, and our assumption that $\Delta$ contains no applications of Axiom~\ref{ax:12awt}, there are no other ways to derive $\widehat{\gamma}$, completing the proof.
\end{proof}
} % end of \arxivonly

\begin{proposition}\label{prop:satt}
If
\begin{equation*}
P \trusts \ltriple{D}{O}{\satt{\ell}}
\end{equation*}
is derivable by repeatedly applying Rules~\ref{rule:mp} and~\ref{rule:ai} to the axioms, the instances of the Local Axiom Schemata in \las{P}, the elements of \tstmts, and the elements of \itrust{P}, then one of the following conditions holds:
\begin{enumerate}
    \item $P \trusts \ltriple{D}{O}{\satt{\widehat{\ell}}}$ is in $\itrust{P}$ for some $\widehat{\ell}\wkbroadens\ell$;\label{satt:itrust} or
    \item $\lassatt{D}{O}\in\las{P}$;\label{satt:lassatt} or
    \item For some $g\in\alllabels$ and  $\widehat{\ell}\wkbroadens\ell$:\label{satt:lasdel}
    \begin{enumerate}
        \item $\lasdel{g}{\satt{\widehat{\ell}}}\in\las{P}$; and
        \item $P \trusts \ltriple{D}{O}{g}$ is derivable;
    \end{enumerate}
    or
    \item $P \trusts \ltriple{D}{O}{\sattdel{\widehat{\ell}}}$ is derivable for some $\widehat{\ell}\wkbroadens\ell$;\label{satt:sattsame} or
    \item For some $\upair{D'}{O'}$ and $\ell'\wkbroadens\ell''\wkbroadens\ell$:\label{satt:sattdel}
    \begin{enumerate}
        \item $P \trusts \ltriple{D'}{O'}{\sattdel{\ell'}}$ is derivable; and
        \item $P \trusts \upair{D'}{O'} \says \ltriple{D}{O}{\satt{\ell''}}$ is derivable
    \end{enumerate}
\end{enumerate}
\end{proposition}

\begin{proof}[Proof of Prop.~\ref{prop:satt}]
Fix a derivation $\Delta$ of $\widehat{\gamma} = P \trusts \ltriple{D}{O}{\satt{\ell}}$ using applications Rules~\ref{rule:mp} and~\ref{rule:ai} to the axioms other than Axiom~\ref{ax:12wt}a, the instances of the Local Axiom Schemata in \las{P}, the elements of \tstmts, and the elements of \itrust{P}.

We work inductively backwards through $\Delta$.  We assume that $\widehat{\gamma} = P \trusts \ltriple{D}{O}{\satt{\widehat{\ell}}}$ for some $\widehat{\ell}\wkbroadens\ell$, and we consider how $\widehat{\gamma}$ was derived in $\Delta$.
\begin{enumerate}
    \item If $\widehat{\gamma}\in\itrust{P}$, then the claimed Condition~\ref{satt:itrust} is satisfied, and we terminate the induction.

    \item If $\widehat{\gamma}$ is obtained by an application of \ref{rule:mp} to an instance of Axiom~\ref{ax:4wt}, then the antecedent of that axiom instantiation includes conjuncts $P \trusts \ltriple{D'}{O'}{\sattdel{\ell'}}$ and $P \trusts \upair{D'}{O'} \says \ltriple{D}{O}{\satt{\ell''}}$ for some \upair{D'}{O'} and $\ell'\wkbroadens\ell''\wkbroadens\widehat{\ell}$.  The conjuncts are thus derivable, satisfying the claimed Condition~\ref{satt:sattdel}, and we terminate the induction.

    \item If $\widehat{\gamma}$ is obtained by an application of \ref{rule:mp} to an instance of Axiom~\ref{ax:7wt}, then the antecedent of that axiom instantiation is $P \trusts \ltriple{D}{O}{\sattdel{\widehat{\ell}}}$; this is derivable, satisfying claimed Condition~\ref{satt:sattsame}, and we terminate the induction.

    \item If $\widehat{\gamma}$ is obtained by an application of \ref{rule:mp} to an instance of Axiom~\ref{ax:9wt}, then the antecedent of that axiom instantiation is $P \trusts \ltriple{D'}{O'}{\sattdel{\ell'}}$ for some \upair{D'}{O'} and $\ell'\wkbroadens\widehat{\ell}\wkbroadens\ell$.  We update $\widehat{\gamma}$ to be this antecedent and continue the induction.
    
    \item If $\widehat{\gamma}$ is obtained by an application of \ref{rule:mp} to an instance of \lasdel{g_1}{g_2}, then we have $g_2 = \satt{\widehat{\ell}}$ and the term $P\trusts\ltriple{D}{O}{g_1}$ must be derivable.  This satisfies the claimed Condition~\ref{satt:lasdel}, and we terminate the induction.
    
    \item If $\widehat{\gamma}$ is obtained by an application of \ref{las:satt}, then $\lassatt{D}{O}\in\las{P}$, satisfying claimed Condition~\ref{satt:lassatt}, and we terminate the induction.
\end{enumerate}

By inspection and our assumption that $\Delta$ contains no applications of Axiom~\ref{ax:12awt}, there are no other ways to derive $\widehat{\gamma}$, completing the proof.
\end{proof}

\begin{theorem}\label{thm:sattdel}
If
\begin{equation*}
P \trusts \ltriple{D}{O}{\sattdel{\ell}}
\end{equation*}
is derivable by repeatedly applying Rules~\ref{rule:mp} and~\ref{rule:ai} to the axioms, the instances of the Local Axiom Schemata in \las{P}, the elements of \tstmts, and the elements of \itrust{P}, then there is a sequence $\{(D_i,O_i,\ell_i)\}_{i=0}^k$ such that the following conditions hold
\begin{enumerate}
    \item $(D_k,O_k,\ell_k) = (D,O,\ell)$;\label{cond:lasttriple} and
    \item One of the following holds:
    \begin{enumerate}
        \item $P\trusts \ltriple{D_0}{O_0}{\sattdel{\ell_0}}\in\itrust{P}$ and $\ell_0\wkbroadens\ell$; or\label{cond:itrust}
        \item $\lassattdel{D_0}{O_0}\in\las{P}$; or\label{cond:lassattdel}
        \item For some $g\in\alllabels$ and some $\ell_0\wkbroadens\ell$:\label{cond:lasdel}
        \begin{enumerate}
            \item $\lasdel{g}{\sattdel{\ell_0}}\in\las{P}$; and
            \item $P \trusts \ltriple{D_0}{O_0}{g}$ is derivable;
        \end{enumerate}
    \end{enumerate}
    and
    \item For $0\leq i< k$:
    \begin{enumerate}
        \item $\ell_i \wkbroadens \ell_{i+1}$;\label{cond:labeli} and
       \item $P\trusts \ltriple{D_{i}}{O_{i}}{\sattdel{{\ell_{i}}}}$\label{cond:trusti} is derivable;
    \end{enumerate}
    and
    \item For $1\leq i < k$, $P$ can derive
    \[
P\trusts \upair{D_i}{O_i} \says \ltriple{D_{i+1}}{O_{i+1}}{\sattdel{\widetilde{\ell_{i+1}}}}
\]\label{cond:linkedtrust} for some $\widetilde{\ell_{i+1}}$ such that $\ell_i \wkbroadens \widetilde{\ell_{i+1}} \wkbroadens \ell_{i+1}$; and
    \item $\upair{D_0}{O_0} = \upair{D_1}{O_1}$\label{cond:initsame}
\end{enumerate}
\end{theorem}

\arxivonly{
\begin{proof}[Proof of Thm.~\ref{thm:sattdel}]
Fix a derivation $\Delta$ of $P \trusts \ltriple{D}{O}{\sattdel{\ell}}$ using applications Rules~\ref{rule:mp} and~\ref{rule:ai} to the axioms other than Axiom~\ref{ax:12wt}a, the instances of the Local Axiom Schemata in \las{P}, the elements of \tstmts, and the elements of \itrust{P}.  Set $(D_k,O_k,\ell_k) = (D,O,\ell)$, where $k$ is the number of applications of Modus Ponens in $\Delta$.  Let
\[
\gamma_k = P \trusts \ltriple{D_k}{O_k}{\sattdel{\ell_k}},
\]
and let $\widehat{\gamma} = \gamma_k$.

We work inductively backwards through $\Delta$.  Assume that $\gamma_{i+1} = P\trusts \ltriple{D_{i+1}}{O_{i+1}}{\sattdel{\ell_{i+1}}}$ has been defined but that $\gamma_i$ has not been defined for some $i$, $0\leq i < k$.  Assume that $\widehat{\gamma}$ has been defined as the statement $P\trusts \ltriple{D_{i+1}}{O_{i+1}}{\sattdel{\widehat{\ell}}}$ for some $\widehat{\ell}\wkbroadens\ell_{i+1}$ and that $\widehat{\gamma}$ appears as (a conjunct in) the precondition of an axiom to which \ref{rule:mp} is applied in $\Delta$.

\begin{enumerate}
    \item\label{case:nomp} If $\widehat{\gamma}$ is not obtained in $\Delta$ by applying \ref{rule:mp}, then it must be either the instantiation of \lassattdel{D_{i+1}}{O_{i+1}} or an element of \itrust{P}; it does not match the form of the other local-axiom schemata, any axiom (without the application of \ref{rule:mp}), the ($\mathsf{s\_stmt}$) elements of \tstmts, or the result of applying Rule~\ref{rule:ai}.  We set $\gamma_i = \widehat{\gamma}$.  If $i>0$, we renumber the $\gamma_j$ by subtracting $i$ from each index; this gives $\upair{D_0}{O_0} = \upair{D_1}{O_1}$.  This puts us in Condition~\ref{cond:itrust} or~\ref{cond:lassattdel}, and we stop the induction.

    \item If $\widehat{\gamma}$ is obtained in $\Delta$ by applying \ref{rule:mp}, then it must be an application of \ref{rule:mp} to Axiom~\ref{ax:3wt}, Axiom~\ref{ax:10wt}, or \lasdel{g_1}{g_2}; by inspection, no other axioms (including instances of local-axiom schemata, and excluding Axiom~\ref{ax:12wt} by the discussion above) can produce a statement including \sattdel{}.  We consider these three cases.
    \begin{enumerate}
        \item If $\widehat{\gamma}$ is obtained by applying \ref{rule:mp} to Axiom~\ref{ax:3wt}, then the precondition of the axiom includes a unique term of the form $P\trusts\ltriple{D'}{O'}{\sattdel{\ell'}}$, where $\ell'\wkbroadens\widehat{\gamma}\wkbroadens\ell_{i+1}$.  We let $(D_i,O_i,\ell_i) = (D',O',\ell')$ and set $\gamma_i = P\trusts\ltriple{D_i}{O_i}{\sattdel{\ell_i}}$.  Note that, in the precondition of the axiom, we also have a unique conjunct of the form $P\trusts\upair{D_i}{O_i}\says\ltriple{D_{i+1}}{O_{i+1}}{\sattdel{\ell''}}$ for some $\ell''$ such that $\ell_i\wkbroadens\ell''\wkbroadens\widehat{\ell}\wkbroadens\ell_{i+1}$.  We set $\widehat{\gamma} = \gamma_i$ and continue the induction.

        \item If $\widehat{\gamma} = P\trusts\ltriple{D_{i+1}}{O_{i+1}}{\sattdel{\widehat{\ell}}}$ is obtained by applying Modus Ponens to Axiom~\ref{ax:10wt}, then the precondition of the axiom includes a unique conjunct of the form $P\trusts\ltriple{D_{i+1}}{O_{i+1}}{\sattdel{\widetilde{\ell}}}$ for some $\widetilde{\ell}\wkbroadens\widehat{\ell}$.  We set $\widehat{\gamma} = P\trusts\ltriple{D_{i+1}}{O_{i+1}}{\sattdel{\widetilde{\ell}}}$ and continue the induction.
        
        \item\label{case:mpdel} If $\widehat{\gamma}$ is obtained by applying Modus Ponens to an instance of \lasdel{g_1}{g_2}, then we must have $g_2 = \sattdel{\widehat{\ell}}$.  Furthermore, $P\trusts\ltriple{D_i}{O_i}{g_1}$ must be derivable.  By assumption, $\widehat{\ell}\wkbroadens\ell_k = \ell$.  We set $\gamma_i = \widehat{\gamma}$.  If $i>0$, we renumber the $\gamma_j$ by subtracting $i$ from each index; this gives $\upair{D_0}{O_0} = \upair{D_1}{O_1}$ (satisfying Condition~\ref{cond:initsame}). This puts us in Condition~\ref{cond:lasdel}, and we terminate the induction.
    \end{enumerate}
\end{enumerate}

The proof cases that terminate the induction identify which of Conditions~\ref{cond:itrust}, \ref{cond:lassattdel}, and~\ref{cond:lasdel} is satisfied.  We consider the other conditions that are also claimed to hold.

Condition~\ref{cond:lasttriple} is satisfied by construction.  For Condition~\ref{cond:labeli}, every update to the definition of $\widehat{\ell}$ ensures that $\widehat{\ell}\wkbroadens\ell_{j}$, where $j$ is the smallest index such that $\ell_i$ has been defined.  When a new $\gamma_i$ is defined, it weakly broadens the then-current value of $\widehat{\ell}$, which (as just noted) weakly broadens $\ell_{i+1}$.

Condition~\ref{cond:trusti} says that $\gamma_i$ is derivable.  Each $\gamma_i$ is defined as the then-current value of $\widehat{\gamma}$; that, in turn, is (a conjunct in) a precondition to which \ref{rule:mp} is applied in $\Delta$.  Thus, the precondition is derivable in $\Delta$, and the conjunct defining $\gamma_i$ may be derived from that precondition (or, as discussed above, directly).

Condition~\ref{cond:linkedtrust} follows from the fact that, when $\gamma_i$ is defined for $1\leq i < k$, it must be by application of~\ref{rule:mp} to Axiom~\ref{ax:3wt}.  In this case, we have as a conjunct of the precondition the term $P\trusts\upair{D_i}{O_i}\says\ltriple{D_{i+1}}{O_{i+1}}{\sattdel{\ell''}}$ for some $\ell''$ such that $\ell_i\wkbroadens\ell''\wkbroadens\widehat{\ell}\wkbroadens\ell_{i+1}$.  From the precondition, $P$ may apply Axiom~\ref{ax:12wt}a and \ref{rule:mp} to obtain the claimed term.

$D_0$ and $O_0$ are defined when the induction is terminated.  This is done by either Case~\ref{case:nomp} or Case~\ref{case:mpdel} in the proof; in either case, Condition~\ref{cond:initsame} of the lemma holds.
\end{proof}
} % end of \arxivonly

%%%%%%%%%%%%%%%%%%%%%%
\section{Saturation Algorithm}
\label{sec:saturation}

We now describe a saturation algorithm, that, given a principal $P$ and a finite 
initial set $\itrust{P}$ of statements, can be used to derive all sattestor-free 
trust statements derivable from using the axioms and rules.  We assume the 
elements of $\itrust{P}$ are 1) $\mathsf{s\_stmt}$s,  and 2) $\mathsf{t\_stmt}$s of the 
form $P\trusts\varphi$ (``$P$-$\!\!\trusts\!$'' statements).  We characterize a finite superset  of $\itrust{P}$
 that the algorithm   uses to construct derivations, and we show that this can be done without sacrificing completeness.
This means that the number of statements the algorithm has to derive is a function of $\itrust{P}$ and $P$'s local axioms, not of the set of all possible statements.  In particular, this guarantees termination even if the set of labels is infinite.

We begin by making some necessary definitions.

\begin{definition}\label{def:saturation}
If  $SF$ is a set of $\mathsf{stmt}$s and/or $\mathsf{form}$s, we let $\mylabels{SF}$ be the set of labels appearing in $SF$, and $\mynames{SF}$ be the set of identities $D,O$ appearing in $SF$. 
Given a label $g$, and a set of labels  $\anylabelset{}$, we define $\blabel{g}$ to be $\ell$ if $g = \ell$ or $g = \sattboth{\ell}$, where we will use $\sattboth{\ell}$ to include both $\satt{\ell}$ and $\sattdel{\ell}$. We then define $\blabels{\anylabelset{}}$ to be $\{\ell = \blabel{g} \mid g \in \anylabelset{}\}$.
Given a set of labels $\anylabelset{}$, we define $\lc{\anylabelset{}}$ to to be the set of all labels $g$ such that $\blabel{g}$ is in  the closure above under $\narrows$ of the set $\blabels{\anylabelset{}}$.

Let $P$ be a principal,  and let $\las{P}$ be the set of axioms local to $P$.  Let $S$ be a collection of $\says$ and $P$-$\!\!\trusts\!$ statements.  
We define the local axiom label closure of $S$, or $\lalc{S,P}$, to be $\lc{\mylabels{S \cup \las{P}}}$. 
If $S_0$ and $S_1$ are two sets of says and $P$ trusts statements we say that $S_1$  is $(S_0,P)$-\emph{limited} if $\mylabels{S_1}   \subseteq \lalc{S_0,P}$ and $\mynames{S_1} \subseteq \mynames{S_0 \cup \las{P}}$.  
For  any two sets of statements $S_0$ and $S_1$, we denote the  maximal $(S_0,P)$-limited subset of $S_1$ by  $\maxlim{S_1,S_0}$.
\end{definition}

We use the following proposition in proving Thm.~\ref{thm:saturation}:
\begin{proposition} \label{prop:finite}
Given  any set $T$ of $\says$ and $P$-$\!\!\trusts\!$ statements
there is a $(T,P)$-limited set $\mathcal{Q}$  such that  a) $P\trusts \ltriple{D}{O}{\ell} \in \mathcal{Q}$ if and only if b) any statement $P\trusts \ltriple{D}{O}{\ell}$ that can be derived from $T$ can be derived using a sequence of statements from $\mathcal{Q}$.
 \end{proposition}

\arxivonly{We present three lemmata that we need to prove Prop.~\ref{prop:finite}.  The first lemma shows that any statement of the form $P\trusts \ltriple{D}{O}{\ell}$ that can be derived from a set $S$ is in $\maxlim{\Sigma,S}$.

\begin{lemma} \label{lem:antecon}
For any set of statements $S$, let $P\trusts \ltriple{D}{O}{\ell}$ be derivable from $S$.  Then $P\trusts \ltriple{D}{O}{\ell}$ is in $\maxlim{\Sigma,S}$.
\end{lemma}
} % end of \arxivonly
\arxivonly{\begin{proof}
This follows from the fact that $P\trusts \ltriple{D}{O}{\ell}$ cannot be derived unless $\upair{D}{O}\says \ltriple{D}{O}{\mathsf{\ell'}}$ is present, where $\ell' \wknarrows \ell$. That statement is not derivable from any other statement, hence it must be in $S$.  Thus $P\trusts \ltriple{D}{O}{\ell}$ is in $\maxlim{\Sigma,S}$
\end{proof}
} % end of \arxivonly

\arxivonly{The next two lemmata are used to help show that any statement $s$ of the form $P\trusts \ltriple{D}{O}{\ell}$ such that  $\mylabels{\{s\}} \subseteq \maxlim{\Sigma,S}$, that is derivable using statements from $\Sigma$, can be derived using statements from $\maxlim{\Sigma,S}$.

\begin{lemma} \label{lem:sayscon}
For any  set of statements $S$, and any statement $s$ of the form $P \trusts \upair{D_1}{O_1} \says \phi$ that is derivable from $S$,  we have $\lc{\phi} \subseteq \lalc{S}{P}$.
\end{lemma}
} % end of \arxivonly
\arxivonly{\begin{proof}
This follows directly from the fact that the  only way such a statement could be derived  from $S$ is if the statement $S = \upair{D_1}{O_1} \says \phi$ appears in $S$.
\end{proof}
} % end of \arxivonly

\arxivonly{
\begin{lemma} \label{lem:conscon}
For any  $S$, let $\mathsf{A}$ be a valid instantiation  of  either axiom A2,  any global axiom A , whose consequent  is  $P \trusts \ltriple{D}{O}{\sattdel{\ell}}$, or a member of the local axiom schema \ref{las:del}.   Then, if the labels of the consequent are in  $ \lalc{S}{P}$, so are the labels of the antecedent.
\end{lemma}
} % end of \arxivonly
\arxivonly{\begin{proof}
The proof follows by inspection of the axioms and local axiom schemata, and by the fact that, by Lemma \ref{lem:sayscon}, for  any derivable statement $P\trusts \upair{D_1}{O_1} \says \phi$, then  $\lc{\phi} \subseteq \lalc{S}{P}$.
\end{proof}
} % end of \arxivonly

\arxivonly{We now turn to the proof of Prop~\ref{prop:finite}.
\begin{proof}[Proof of Prop.~\ref{prop:finite}]

The b) implies a) direction  follows directly from Lemma \ref{lem:antecon}.

 We prove the a) implies b) direction by assuming a  $\mathcal{Q}$ that satisfies a), and then show that when unnecessary statements are removed it satisfies b). 
 To prove b), we begin by proving  that any sattestor-free statement $P \trusts \ltriple{D}{O}{\ell}$ that is derivable from $\mathcal{S}$  is $(T,P)$-limited.  .Next we need to show that if there is a derivation $Der$ of a $(T,P)$-limited statement $P \trusts \ltriple{D}{O}{\ell}$, there is a $(T,P)$-limited derivation of $P \trusts \ltriple{D}{O}{\ell}$.

 We first prove that $\mynames{\Der} \subseteq \mynames{T  \cup \las{P}}$.  
We note that every  SATA $\upair{D}{O}$ that appears in the the conclusion of an axiom also appears in the antecedent, except for instances  of the schemata \ref{las:satt} and \ref{las:sattdel}.   However, in order for those local axioms to prove anything except more instances of  $P \trusts \ltriple{D'}{O'}{\sattboth{\ell'}}$, we need $P$ to trust either that  $\upair{D'}{O'}$ has said something, or  somebody has said something about $\upair{D'}{O'}$.   In that case  $\upair{D'}{O'} \in \mynames{T\cup\las{P}}$. Otherwise, we can remove this unnecessary  $P \trusts \ltriple{D'}{O'}{\sattboth{\ell'}}$ statements.

The proof that $\mylabels{Der}   \subseteq \lalc{T}{P}$ is similar but is complicated by the use of the narrowing relation, so it takes a little more work.  Let  $s'= P \trusts \ltriple{D'}{O'}{g}$ be the last statement in $Der$  such that $\blabel{g} \not \in \lalc{T}{P}$. We note that  $s$ cannot be used as input to an instantiation of an antecedent Axiom A2, or any axiom  whose  consequent is of the form $P \trusts \ltriple{D}{O}{\sattboth{\ell'}}$ and contains $P \trusts \ltriple{D'}{O'}{\sattboth{\ell''}}$ in the antecedent, because by Lemma \ref{lem:conscon}, this would result in a consequent $s = P \trusts \ltriple{D}{O}{\sattboth{\hat{\ell}}}$ such that $\ell \wkbroadens \hat{\ell}$, so that $\hat{\ell} \not \in \lalc{T}{P}$, contradicting the assumption.   Thus it can be removed without affecting the correctness of the derivation. 
 \end{proof}
} % end of \arxivonly
 
Our saturation algorithm is as follows:
Let $T$ be a set of non-negated atomic statements. 
Let $S = \maxlim{\mathsf{All},T}$, where $\mathsf{All}$ is the set of all possible non-negated atomic statements.
\begin{enumerate}
\item Initialize: $\present= T$, $\past := \emptyset$,  $\rest :=  S \setminus \present$.
\item While $\present \ne \past$, do:
\item \hspace{10pt} $U:= \present$
\item \hspace{10pt} $V := \rest$ 
\item \hspace{10pt} $U := U \cup Q$, where $Q$ is the set of all elements of $V$ derivable from $U$ in one step.
\item \hspace{10pt} $\past := \present$,  $\present := U$, $\rest := S \setminus \present$
\item End while
\item $\resulta = \present$
\item $\resultb = \sattfree{\present}$

\item  Return $\resulta$, $\resultb$
 \end{enumerate}

 We give an example of the saturation algorithm at work below.
 Step 0 gives the statements $P$ trusts initially, while Steps 1 and 2 gives the statements $P$ derives at each respective step.
 
 \paragraph{Step 0}
\begin{itemize}
\item[01]   $P\trusts \ltriple{D_1}{O_1}{\satt{\mathsf{a}}}$ 
 \item[02]   $P\trusts \upair{D_1}{O_1}\says \ltriple{D_2}{O_2}{\mathsf{a:b}}$
 \item[03]  $\upair{D_2}{O_2}\says \ltriple{D_2}{O_2}{\mathsf{a:b}}$  
\end{itemize}

%\newpage

\paragraph{Step 1}
\begin{itemize}
\item[11] $P\trusts \ltriple{D_1}{O_1}{\satt{\mathsf{a:b}}}$ (from statement 01 via Axiom~\ref{ax:9wt})
\item[12]$P\trusts \upair{D_1}{O_1}\says \ltriple{D_2}{O_2}{\mathsf{a:b}}$ (from statement   02 via Axiom~\ref{ax:11wt})
\item[13] $P\trusts \upair{D_2}{O_2} \says  \ltriple{D_2}{O_2}{\mathsf{a:b}}$ (from statement  03  via Axiom~\ref{ax:11wt})
\end{itemize}

\paragraph{Step 2}
\begin{itemize}
\item[21] $P\trusts \ltriple{D_2}{O_2}{\mathsf{a:b}}$ (from statements  01, 12, and 13  via Axiom~\ref{ax:2wt})
\end{itemize}

\begin{theorem}\label{thm:saturation}
The saturation algorithm is sound (any  statement in the
saturation set $\resultb$ is derivable from $T$), complete (if a sattestation-free 
statement is derivable from $T$, it appears in $\resultb$), 
and terminating,  with  $\mylabels{\resultb} \subseteq \lc{T}$ and $\mynames{\resultb} \subseteq \mynames{T}$. 
\end{theorem}

\begin{proof} 
Soundness  follows directly from the
definition of derivation. Termination follows from the finiteness of $S$, and $\mylabels{\resultb} \in \lc{T}$ and $\mynames{\resultb} \in \mynames{T}$ follows from the fact that same is true of $S$.

For completeness,   Prop.~\ref{prop:finite} gives us that any sattestation-free statement in $S$ derivable from $T$ is derivable using statements in $S$,  and that any sattestation-free statement derivable from $T$ is itself in $S$.  All that remains is to show that such a derivation can be constructed using the elements of $\resulta$.

We first note  that all statements added by the algorithm are non-negated, so addition of a statement $s$ to $\present$ does not invalidate any statement that could be derived from the statements in $\present$ alone.  Thus all we need to show is that $\resulta$ contains every  statement derivable from $T$ via statements in $S$.  We prove this by induction on  $\dlen{s}$, the minimum length of any  derivation $\Der$ of $s$ from $T$ using only statements from $S$, where the length of $\Der$ is defined to be the length of the longest path in $\Der$, considered as a graph.
First, if $\dlen{s} = 0$, then $s \in T \subseteq \resulta$. Suppose now that if $\dlen{s} < n$ then $s \in \resulta$.  Suppose that there is a statement $s$ such that $\dlen{s} = n$ and $s \not \in \resulta$.  That means that,  if $\Der$ is a derivation of $s$ from $T$ using only statements from $S$ then at least one  statement $s'$ used  in  $\Der$ is not in $\resulta$.  However, by construction, $\dlen{s'} < n$, contradicting the induction hypotheses. 
\end{proof}

\section{Conclusion}
\label{sec:conclusion}

We have set out a logic of sattestation that is a foundation for reasoning about
contextual roots of trust for internet domains if they are
SATAs. Previous work has presented logics for reasoning about
contextual trust in the sense of trust in statements of any kind by a
principal on a topic about which the principal is trusted to have
expertise~\cite{inferring-trust}. But we believe this is the first
logic to support reasoning about the identity and authentication
of internet addresses grounded in the contextual trust of the
attesting principal. Contextual trust is based in sattestations,
an implemented technology for creating, distributing, and validating
credentials issued by SATAs contextually trusted to issue
them~\cite{satas-wpes21}.
Sattestations have been previously implemented, but we introduce
a logic to reason about sattestations in the rich practical enviroments
where contextual trust may depend on organizational or social structure
and where trust might be delegated within such structures. 
Despite the richness and complexity of our logic, we prove that
it is sound with respect to the set of sattestation credentials
trusted by a principal at a given time. We also prove results
about what trust statements can be derived from a set of initial
trust assumptions.

%%% Local Variables: 
%%% mode: latex 
%%% TeX-master: "../paper"
%%% End:          

\bibliographystyle{IEEEtran}
\bibliography{jsm24csf-full}

\end{document}